\documentclass[pdflatex,sn-basic]{sn-jnl}% Basic Springer Nature Reference Style/Chemistry Reference Style
\usepackage{graphicx}%
\usepackage{multirow}%
\usepackage{amsmath,amssymb,amsfonts}%
\usepackage{amsthm}%
\usepackage{mathrsfs}%
\usepackage[title]{appendix}%
\usepackage{xcolor}%
\usepackage{textcomp}%
\usepackage{manyfoot}%
\usepackage{booktabs}%
\usepackage{algorithm}%
\usepackage{algorithmicx}%
\usepackage{algpseudocode}%
\usepackage{listings}%
\usepackage{mdframed}
\theoremstyle{thmstyleone}%
\newtheorem{theorem}{Theorem}%  meant for continuous numbers
\newtheorem{proposition}{Proposition}% to get separate numbers for theorem and proposition etc.

\newtheorem{lemma}{Lemma}

\theoremstyle{thmstyletwo}%
\newtheorem{remark}{Remark}%

\theoremstyle{thmstylethree}%
\newtheorem{definition}{Definition}%

\begin{document}

\title[Polynomial-time completion of phylogenetic tree sets]{Polynomial-time completion of phylogenetic tree sets}


\author[1]{\fnm{Aleksandr} \sur{Koshkarov}}\email{Aleksandr.Koshkarov@USherbrooke.ca}

\author*[1]{\fnm{Nadia} \sur{Tahiri}}\email{Nadia.Tahiri@USerbrooke.ca}

\affil[1]{\orgdiv{Department of Computer Science}, \orgname{University of Sherbrooke}, \orgaddress{\street{boul. de l'Université}, \city{Sherbrooke}, \postcode{J1K 2R1}, \state{Quebec}, \country{Canada}}}


\abstract{Comparative analyses of phylogenetic trees typically require identical taxon sets, however, in practice, trees often include distinct but overlapping taxa. Pruning non-shared leaves discards phylogenetic signal, whereas tree completion can preserve both taxa and branch-length information. This work introduces a polynomial-time algorithm for set-wide completion of phylogenetic trees with partial taxon overlap. The proposed method identifies and extracts maximal completion subtrees that frequently appear across the source trees and constructs a weighted majority-rule consensus. Branch lengths are scaled using rates derived from common leaves. Each consensus subtree is inserted at the position that minimizes the quadratic distance error measured against information from the source trees, with candidate positions restricted to the original branches of the target tree. We demonstrate that the algorithm runs in polynomial time and preserves distances among the original taxa, yielding a unique completion that is order-independent with respect to the processing order of target trees. An experimental evaluation on amphibians, mammals, sharks, and squamates shows that the proposed method consistently achieves the lowest distance to the subset reference trees across subsets among all methods, in both topology and branch lengths.
   
An open-source Python implementation of the proposed algorithm and the biological datasets utilized in this study are publicly available at: \url{https://github.com/tahiri-lab/overlap-treeset-completion/}.
}




\maketitle

\section{Introduction}\label{sec:intro}

A phylogenetic tree is a connected acyclic graph that encodes a hypothesis about the evolutionary relationships among a set of taxa. The leaves of the phylogenetic tree are taxa, the edges represent lineages, and the internal nodes connect lineages. In a rooted tree, the root sets the direction from ancestors to descendants, internal nodes correspond to inferred most recent common ancestors, and each split represents a divergence event. A clade is any ancestor together with all of its descendants. Phylogenetic trees are typically inferred based on aligned nucleotide sequences of selected taxa, whereas other data can also be utilized for this purpose. Different tree reconstruction methods and datasets may yield different trees for the same or overlapping taxa because of statistical uncertainty, model differences, and biological processes (e.g., incomplete lineage sorting, horizontal gene transfer, and hybridization) \citep{zou2024common}. To reconcile conflicts among trees, one or more summary trees can be constructed from the input trees. Two well-known approaches for this purpose are consensus trees (same leaf set) and supertrees (overlapping leaf sets).

Consensus methods synthesize similarities among phylogenetic trees that share identical taxon sets. There are various well-established methods for constructing consensus trees. Basic approaches include, but are not limited to, the following. The strict consensus method \citep{mcmorris1983view} retains only clades present in all input trees, often resulting in unresolved polytomies if conflicts exist among the input trees. The majority-rule consensus \citep{dong2010majority,margush1981consensusn} is a more flexible approach that includes clades present in more than 50\% of the input trees. Adams consensus \citep{adams1972consensus} focuses on preserving ancestral-descendant relationships, whereas the greedy consensus \citep{bryant2003classification,degnan2009properties} prioritizes frequently occurring clades. Other popular consensus methods are loose consensus \citep{jansson2016improved}, frequency difference consensus \citep{goloboff2003improvements,jansson2024faster}, local consensus \citep{jansson2016minimal,kannan1998computing}, and R* consensus \citep{bryant2003classification,jansson2016faster}.

Supertree construction involves the integration of all trees within a given set, consisting of overlapping trees, into a supertree. Various methods have been developed to handle overlapping taxa and resolve conflicts between input trees. Some widely used supertree methods include matrix representation approaches, such as matrix representation with parsimony \citep{baum1992combining,ragan1992phylogenetic}, matrix representation with likelihood \citep{akanni2015implementing,nguyen2012mrl,steel2000parsimony}, and matrix representation with compatibility \citep{pisani2002matrix,purvis1995modification}. Other notable techniques include average consensus \citep{lapointe1997average}, SuperFine \citep{swenson2012superfine}, most similar supertree \citep{creevey2004does}, FastRFS \citep{vachaspati2017fastrfs}, majority-rule supertrees \citep{kupczok2011split}, and MinCut supertree methods \citep{page2002modified,semple2000supertree}. These methods integrate information from overlapping trees, considering differences in taxon sets and resolving conflicts. Additionally, quartet-based supertree methods \citep{piaggio2004quartet} construct supertrees by combining smaller groups of four-taxon trees (quartets) into a larger tree. These methods are particularly applicable for datasets with sparse overlap, where only limited taxon information is shared across trees. More recent techniques such as the spectral cluster supertree \citep{mcarthur2024spectral} have incorporated spectral clustering approaches, improving both the scalability and accuracy of supertree construction.

Although individual methods exist for consensus trees, supertree construction, and tree completion \citep{bansal2020linear,koshkarov2024novel,li2024comparison}, it is important to note that the set-wide completion of collections of overlapping phylogenetic trees has not yet been explicitly addressed in the scientific literature. To the best of our knowledge, this represents a gap in the literature. This study addresses this gap by proposing a novel approach to phylogenetic tree set completion.

\textbf{Problem statement.} Let $\mathcal{T}=\{T_1,\dots,T_n\}$ be a collection of rooted binary phylogenetic trees with strictly positive branch lengths and overlapping leaf sets, and let $L(\mathcal{T})=\bigcup_{i=1}^n L(T_i)$. For each target tree $T_i$, the goal is to construct a completed tree $T_i^{\uplus}$ on $L(\mathcal{T})$ such that the induced subtree on the original taxa is preserved. A trivial strategy that attaches missing taxa arbitrarily ignores consistent topological patterns and branch length information present in the source trees and can produce completions that are incompatible with relationships frequently observed across the tree set.

To address this problem, each insertion is chosen by minimizing a quadratic discrepancy objective that compares distances from common leaves in $T_i$ to a candidate attachment point with the corresponding distances inferred from the source trees. The minimum is taken over all positions along the original branches of $T_i$. The formal expression of this objective function and its minimization are given in Equation~\ref{eq:obj_func} and Section~\ref{sec:methods}. The procedure is applied iteratively until all missing taxa are inserted, while the original topology on $L(T_i)$ remains unchanged.

\textbf{Our contributions.} This paper makes several key contributions that advance the state of phylogenetic tree set completion. (1) We define the phylogenetic tree set completion problem for collections of trees with overlapping taxa and introduce maximal completion subtrees and their consensus versions as the core structures that support phylogenetic tree set completion. To our knowledge, existing work does not address the problem of completing sets of phylogenetic trees for collections of trees with different yet overlapping sets of taxa. (2) We develop an iterative, locally optimal subtree selection strategy that maintains the set of uncovered taxa, extracts candidates whose leaves lie in this set, filters them by cross-tree frequency, groups candidates by identical leaf sets, and chooses the group that maximizes coverage at each step. (3) We extend classical consensus tree constructions by weighting splits according to the overlap between each source and the current target tree. These weights determine both majority-rule topology and branch lengths via weighted averaging. (4) We adjust and extend the tree completion algorithm based on common leaves introduced in \citep{koshkarov2024novel} from pairwise completion to a set-wide completion strategy. (5) We analyze the computational core of the subtree coverage step via parameterized complexity. By formulating the optimal coverage variant as an exact-cover problem, we show NP-completeness and prove fixed-parameter tractability when parameterized by the number of selected subtrees and the maximum candidate subtree size. (6) We propose an optional deterministic multifurcation resolution that refines completed trees by inserting only zero-length internal branches, producing outputs without multifurcations while preserving the leaf set and all pairwise leaf-to-leaf distances. (7) Finally, we provide an experimental evaluation on biological datasets (amphibians, mammals, sharks, and squamates), comparing against baseline methods using the normalized Robinson-Foulds distance and the normalized Branch Score Distance relative to subset reference trees. Our approach yields a unique completion for each input tree and is order-independent with respect to the processing order of target trees. We also provide a polynomial-time complexity analysis for the entire phylogenetic tree set completion algorithm.

This paper is organized as follows. Section~\ref{sec:notation} introduces notation and preliminaries. Section~\ref{sec:methods} describes the proposed algorithm, detailing the selection of maximal completion subtrees, the construction of consensus subtrees, and their integration into input trees. Section~\ref{sec:properties} outlines several theoretical properties of the proposed phylogenetic tree set completion algorithm, including its local optimality, correctness, order invariance, and polynomial-time bounds. Subsection~\ref{sec:param} is devoted to a parameterized complexity analysis of the subtree selection step. Subsection~\ref{sec:res-multifurc} describes an optional post-processing refinement that deterministically resolves multifurcations in completed trees. Section~\ref{sec:experiments} presents the experimental evaluation of the proposed algorithm on biological data of amphibians, mammals, sharks, and squamates. The conclusion summarizes the findings and outlines potential directions for future research.

\section{Notation and preliminaries}\label{sec:notation}

Given a rooted binary phylogenetic tree $T_i$, let $L(T_i)$ denote its set of taxa (leaves), $V(T_i)$ denote its set of nodes, $E(T_i)$ denote its set of branches (edges), and $|e|$ denote the length of a branch $e=(u,w) \in E(T_i)$. Let $\mathcal{T}=\{T_1,\dots,T_n\}$ be a set of phylogenetic trees, where each tree is defined on an overlapping subset of taxa and all input tree branch lengths are strictly positive. The union of all taxa across $\mathcal{T}$ is denoted by $L(\mathcal{T}) = \bigcup_{T_i \in \mathcal{T}} L(T_i)$. A \textit{target tree} $T_i \in \mathcal{T}$ is the specific tree under consideration for completion. The remaining trees in $\mathcal{T} \setminus \{T_i\}$ are referred to as \textit{source trees}. Let $L^{\uplus}(T_i) = L(\mathcal T)\setminus L(T_i)$ denote the taxa missing from $T_i$. Let $r_{T_i}$ denote the root of $T_i$. The terms \textit{phylogenetic tree} and \textit{tree} are used interchangeably in this article. We use $\mathbf{1}_{\{\cdot\}}$ for indicator functions.

\begin{definition}[Overlapping phylogenetic tree set]
An overlapping phylogenetic tree set $\mathcal{T}$ is a collection of phylogenetic trees $\{T_1, T_2, \dots, T_n\}$, where each tree $T_i$ has a leaf set $L(T_i) \subseteq L(\mathcal{T})$, and there exists at least one leaf common to all trees in $\mathcal{T}$. Formally, the set of common leaves is defined as $CL(\mathcal{T}) = \bigcap_{T_i \in \mathcal{T}} L(T_i)$, and it is required that $CL(\mathcal{T}) \neq \emptyset$.
More generally, for a tree $T_i \in \mathcal{T}$ and a subset $\mathcal{T}' \subseteq \mathcal{T}\setminus\{T_i\}$, the set of common leaves between $T_i$ and $\mathcal{T}'$ is defined as $CL(T_i,\mathcal{T}') = L(T_i) \cap \bigcap_{T \in \mathcal{T}'} L(T)$.
\end{definition}

Throughout this paper, we assume a stronger condition in which each input overlapping phylogenetic tree set $\mathcal{T}$ has at least two leaves in common with all trees. Thus, we specifically require that $|CL(\mathcal{T})| \geq 2$ because the proposed tree completion method uses branch length scaling and determines attachment points from pairwise relationships among common leaves, which are not defined if there is only one common leaf.

\begin{definition}[Distinct-leaf subtree]
A distinct-leaf subtree for a target tree $T_i$ is a rooted subtree $S$ extracted from a source tree $T_j \in \mathcal{T} \setminus \{T_i\}$, such that $L(S)$ contains only leaves missing from $T_i$, i.e., $L(S) \subseteq L(\mathcal{T}) \setminus L(T_i)$. A subtree $S \subseteq T_j$ refers to the rooted subtree induced by a node of $T_j$ and all its descendants (an internal node if $|L(S)| > 1$, or a single leaf if $|L(S)| = 1$).
\end{definition}

\begin{definition}[Maximal completion subtree] \label{def:mcs}
Let $T_i \in \mathcal{T}$ be a target tree. Let $p \in (0,1]$ be the current frequency threshold, and let $U \subseteq L(\mathcal{T}) \setminus L(T_i)$ be the current uncovered set. A distinct-leaf subtree $S$ for $T_i$ with leaf set $L(S)\subseteq U$ is a maximal completion subtree (MCS) at the current threshold $p$ (relative to $U$) if $\bigl|\{\,T_j \in \mathcal{T}\setminus\{T_i\}\ \mid\ \exists\, S' \subseteq T_j,\; L(S')=L(S)\,\}\bigr| \ge p \cdot \bigl|\mathcal{T}\setminus\{T_i\}\bigr|$ and for every distinct-leaf subtree $S''$ for $T_i$ with $L(S'')\subseteq U$ that satisfies $\bigl|\{\,T_j \in \mathcal{T}\setminus\{T_i\}\ \mid\ \exists\, S' \subseteq T_j,\; L(S')=L(S'')\,\}\bigr| \ge p \cdot \bigl|\mathcal{T}\setminus\{T_i\}\bigr|$, the inequality $|L(S)| \ge |L(S'')|$ holds. Thus, maximality is defined with respect to maximum leaf coverage among the distinct-leaf subtrees satisfying the frequency condition within the current uncovered set $U$ at the current threshold $p$.
\end{definition}

The \textit{root node} of the subtree $S$ in $T_j$, denoted by $r_S$, is defined as the lowest common ancestor (LCA) of the leaves $L(S)$ if $|L(S)|>1$, or the single leaf in $L(S)$ if $|L(S)|=1$. The \textit{connecting branch} of $S$ is the branch that connects $r_S$ to its parent node in $T_j$. The parent node of $r_S$ in $T_j$ is referred to as the \textit{attachment node} of $S$ and is denoted by $att(S)$.

\begin{definition}[Completed tree]
For a target tree $T_i \in \mathcal{T}$, a completed tree $T_i^{\uplus}$ is a rooted phylogenetic tree (not necessarily binary) with $L(T_i^{\uplus}) = L(\mathcal{T})$ such that $(T_i^{\uplus})|_{L(T_i)} \cong T_i$ as rooted leaf-labeled trees.
\end{definition}

\begin{definition}[Distance between two nodes]
The distance between any two nodes $u$ and $w$ in a phylogenetic tree $T_i$, denoted as $d^{(T_i)}(u, w)$, is the sum of the branch lengths along the unique path connecting $u$ and $w$ in $T_i$.
\end{definition}

\section{Methods} \label{sec:methods}

The phylogenetic tree set completion method completes each target tree $T_i \in \mathcal{T}$ by iteratively extracting maximal completion subtrees from overlapping source trees, constructing consensus MCS, and inserting them into $T_i$. The resulting completed tree $T_i^{\uplus}$ contains all taxa in $L(\mathcal{T})$, preserving both topological and branch length information.

\subsection{Phylogenetic tree set completion algorithm}

At a high level, the proposed method completes each target tree $T_i$ based on a fixed set of source trees. It maintains the set $U$ of taxa missing from $T_i$ and proceeds in iterations. In each iteration, it (i) identifies a group of maximal completion subtrees whose leaves lie in $U$ and are supported across source trees, using an adaptive frequency threshold, (ii) summarizes that group into a consensus MCS via a weighted majority-rule procedure that favors sources with greater overlap with $T_i$, and (iii) inserts the consensus subtree into $T_i$ at an attachment point chosen by minimizing a quadratic distance discrepancy objective function. The uncovered set $U$ is updated, and the process repeats until all missing taxa are placed. Applying this procedure to every tree yields the completed collection $\mathcal{T}^{\uplus}$.

The insertion of missing taxa $L^{\uplus}(T_i)$ into $T_i$ is based on the identification of MCS within the collection $\mathcal{T} \setminus \{T_i\}$. The approach proceeds as follows (see also Figure~\ref{fig:illust}).

\paragraph{Selection of MCS.} The iterative selection of MCS includes the following steps. A dynamic set of uncovered taxa, denoted as $U$, is initially set to $L^{\uplus}(T_i)$. In each iteration $t$, candidate subtrees $S\subseteq T_j$ with $L(S)\subseteq U$ are extracted from each $T_j\in\mathcal{T}\setminus\{T_i\}$. This is done by identifying nodes in $T_j$ whose descendant leaves are entirely contained in $U$. The corresponding induced rooted subtrees are then collected into the set of candidate subtrees $\mathcal{S}_t$.

Subsequently, each subtree $S \in \mathcal{S}_t$ is assigned a frequency $P(L(S))$, which quantifies the occurrence of subtrees with the same set of leaves in $\mathcal{T}\setminus\{T_i\}$ (Equation~\ref{eq:proportion_subtrees}).
\begin{equation}
   P(L(S)) = \frac{\bigl|\{T_j \in \mathcal{T}\setminus\{T_i\}: \;\exists\, S'\subseteq T_j, \; L(S')=L(S)\, \}\bigr|}{|\mathcal{T}\setminus\{T_i\}|}. 
   \label{eq:proportion_subtrees}
\end{equation}

In Equation~\ref{eq:proportion_subtrees}, each source tree contributes at most one to the numerator for a fixed leaf set, even if that tree contains multiple candidate subtrees with the same leaf set.

Subtrees with a frequency $P(L(S)) < p$ are excluded in this step. The threshold is initialized at $p=0.5$, corresponding to a majority-rule setting. The retained subtrees are $\mathcal{S}'_t = \{\,S \in \mathcal{S}_t: P(L(S)) \ge p\}$.

This filtering step prioritizes subtrees that appear in at least a fraction $p$ of the source trees. If no candidate subtrees satisfy the frequency threshold $p$, a fallback mechanism is applied. In this case, $p$ is decremented in fixed steps of 0.05 until at least one candidate passes the filter. This adaptive relaxation ensures that rare but necessary taxa are covered, with priority still given to highly supported subtrees.

The subtrees in $\mathcal{S}'_t$ that satisfy the threshold criterion are then grouped according to their sets of leaves, such that $\mathbb{G}_t = \{ G \subseteq \mathcal{S}'_t: \forall S,\,S' \in G,\; L(S) = L(S') \}$, where each set $G \in \mathbb{G}_t$ is a subset of $\mathcal{S}'_t$ containing all subtrees with an identical set of leaves.

Let $L(G)$ denote the common leaf set of all subtrees in group $G \in \mathbb{G}_t$. Although each subtree $S \in G$ satisfies $L(S) \subseteq U$ by construction, the uncovered taxa $U$ may appear in multiple groups. The group that maximizes the leaf coverage is selected, that is, choose $G'_t \in \underset{G \in \mathbb{G}_t}{\mathrm{argmax}}\bigl|L(G) \bigr|$ and set $\mathcal{M}_t \gets G'_t$. The selected group $\mathcal{M}_t$ corresponds to a maximal completion subtree at the current threshold $p$ relative to $U$ in the sense of Definition~\ref{def:mcs}. Since all subtrees in $\mathcal{M}_t$ have the same leaf set, $L(\mathcal{M}_t)$ denotes this common leaf set, equal to $L(G'_t)$. In cases where multiple groups achieve the same maximum coverage, lexicographic ordering of leaf sets is used to ensure consistency and invariance.

The uncovered taxa set $U$ is updated as follows (Equation~\ref{eq:u_update}).
\begin{equation}
    U \gets U \setminus L(\mathcal{M}_t). \label{eq:u_update}
\end{equation}

The construction of maximal completion subtrees continues until $U = \emptyset$.

\begin{figure}
    \centering
    \includegraphics[width=1\linewidth]{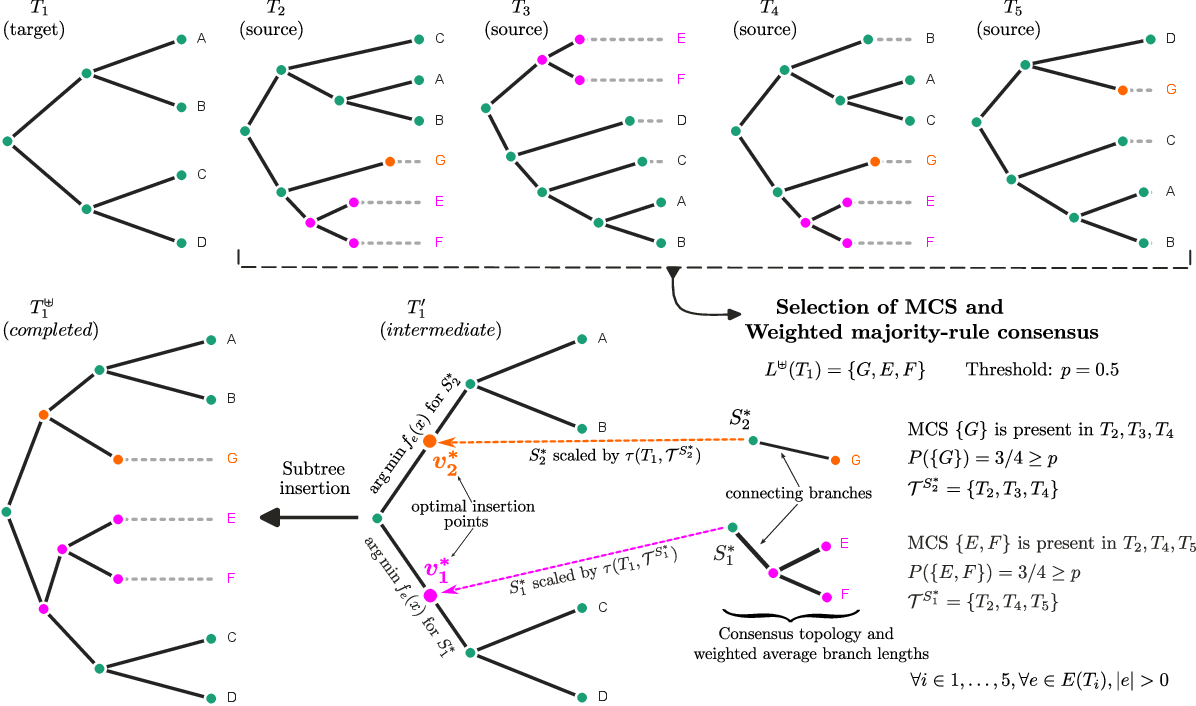}
    \caption{Schematic illustration of phylogenetic tree set completion on five input trees (i.e., trees $T_1$ to $T_5$). Tree $T_1$ is the target tree and trees $\{T_2,\dots,T_5\}$ are source trees, with $L^{\uplus}(T_1)=\{E,F,G\}$. All input trees have strictly positive branch lengths, i.e., $\forall i\in\{1,\dots,5\},\ \forall e\in E(T_i),\ |e|>0$. MCS with leaf sets $\{E,F\}$ and $\{G\}$ are identified in 3/4 source trees, satisfying the threshold $p=0.5$. For each MCS leaf set, a weighted majority-rule consensus subtree ($S_1^*$ and $S_2^*$) is constructed with weighted average branch lengths. In the working copy $T_1'$, optimal insertion points $v_1^*$ and $v_2^*$ are selected by minimizing the objective function $f_e(x)$ (Equation~\ref{eq:obj_func}). The consensus subtrees are then scaled by the corresponding adjustment rate (Equation~\ref{eq:adj_rate}) and inserted at $v_1^*, v_2^*$ to obtain the completed tree $T_1^{\uplus}$. Repeating this procedure for each target tree results in $\mathcal{T}^{\uplus}$.}
    \label{fig:illust}
\end{figure}

\paragraph{Weighted majority-rule consensus.} For each set $\mathcal{M}_t$ of maximal completion subtrees selected at iteration $t$, a consensus maximal completion subtree, denoted by $S^*$, is constructed by applying a weighted majority-rule procedure. Define the contributing source tree index set as $\mathcal{I}_t = \bigl\{ j \in \{1,\dots,n\}\setminus\{i\}:\exists\, S_j \subseteq T_j, \; L(S_j) = L(\mathcal{M}_t) \bigr\}$. For a fixed leaf set in a rooted tree, such a subtree $S_j$, rooted at the LCA of $L(\mathcal{M}_t)$, is unique if it exists, and $\mathcal{M}_t = \{S_j : j\in\mathcal{I}_t\}$ (each $S_j \in \mathcal{M}_t$ originates from a source tree $T_j \in \mathcal{T} \setminus \{T_i\}$).

In order to prioritize the source trees with greater overlap with $T_i$, a weight $w_j$ is assigned to each source tree $T_j$ (Equation~\ref{eq:weights}), with $w_j>0$ for all $j\in\mathcal{I}_t$ under the assumption $|CL(\mathcal{T})|\ge 2$.
\begin{equation}
    w_j = \frac{|L(T_j) \cap L(T_i)|}{|L(T_i)|}. \label{eq:weights}
\end{equation}

Each subtree $S_j$ determines a set of splits $\mathcal{SP}_j$, which represent all bipartitions of the leaf set $L(\mathcal{M}_t)$ induced by the internal branches of $S_j$. Specifically, removing an internal branch from $S_j$ disconnects it into two subtrees, and the corresponding split is the pair of leaf sets on either side of that branch. Each such split represents a grouping of taxa that share a common ancestor. Since $S_j$ is derived from $T_j$, each split $sp \in \mathcal{SP}_j$ is associated with the same weight $w_j$ as its corresponding source tree $T_j$. For $sp\in \bigcup_{j\in\mathcal{I}_t}\mathcal{SP}_j$, let $e_j(sp)$ denote the internal branch of $S_j$ that induces $sp$ (defined only if $sp\in\mathcal{SP}_j$), and let $|e_j(sp)|$ be the length of that branch in $S_j$.
The weighted frequency of $sp$, i.e., the sum of the weights of subtrees in which it occurs divided by the total weight of all subtrees, is computed by Equation~\ref{eq:freq_split}.
\begin{equation}
     f(sp) = \frac{\sum\limits_{j\in\mathcal{I}_t} w_j \cdot \mathbf{1}_{\{sp\in\mathcal{SP}_j\}}}{\sum\limits_{j\in\mathcal{I}_t} w_j}.\label{eq:freq_split}
\end{equation}

If $|L(\mathcal{M}_t)| > 1$, a rooted majority-rule consensus subtree $S^*$ with branch lengths is constructed by including all splits $sp$ with $f(sp) > 0.5$. For each split $sp$, the corresponding branch $e^*(sp)$ in $S^*$ is assigned a length $|e^*(sp)|$, computed as the weighted average of the branch lengths in the contributing source trees (Equation~\ref{eq:cons_length}).
\begin{equation}
   |e^*(sp)| = \frac{\sum\limits_{j\in\mathcal{I}_t} w_j\cdot \mathbf{1}_{\{sp\in \mathcal{SP}_j\}}\, |e_j(sp)|}{\sum\limits_{j\in\mathcal{I}_t} w_j\cdot \mathbf{1}_{\{sp\in \mathcal{SP}_j\}}}.
   \label{eq:cons_length}
\end{equation}

To finalize the construction of $S^*$, the connecting branch from the attachment placeholder (node) to the root $r_{S^*}$ is included as an integral component, with its length computed using Equation~\ref{eq:cons_length}. In this case, the branch lengths are taken as $|(att(S_j), r_{S_j})|$, with the indicator function equal to 1 for each $j \in \mathcal{I}_t$ where this branch exists, and 0 otherwise.

If $|L(\mathcal{M}_t)|=1$, the induced split sets $\mathcal{SP}_j$ are empty, and the consensus step returns a single-leaf subtree. In this case, split processing is skipped, and only the connecting branch is included, with its length computed by Equation~\ref{eq:cons_length}.

\paragraph{Subtree insertion.} Once each consensus maximal completion subtree $S^*$ is constructed, it is inserted into the target tree $T_i$. This integration follows an updated version of the tree completion algorithm, which is based on the common leaves approach \citep{koshkarov2024novel}. This tree completion method ensures that the induced subtree on $L(T_i)$ and all pairwise distances among leaves in $L(T_i)$ remain unchanged. The completion process is described as follows.

The insertion is performed on an intermediate tree $T'_i$, a copy of the target tree $T_i$ that is incrementally updated as consensus subtrees are inserted. It represents the current state of the tree before inserting the next consensus maximal completion subtree $S^*$. With $T_i$ left unchanged, all weight and distance scaling calculations always refer to the original topology, while $T'_i$ accumulates only the state of prior insertions. This separation makes it possible to place each new subtree $S^*$ relative to the current partial completion without altering $T_i$, and guarantees that subsequent MCS extractions are always performed on the same unmodified input trees.

Let $\mathcal{T}^{S^*}$ denote the set of source trees that contributed to the construction of $S^*$.

The insertion process begins by identifying the common leaves $CL(T_i, \mathcal{T}^{S^*})$, which are used to compute distance scaling factors. Formally, for a target tree $T_i$ and a set of source trees $\mathcal{T}^{S^*}$, the set of their common leaves is $CL(T_i, \mathcal{T}^{S^*})= \bigl\{l\in L(T_i):l \in L(T_j) \; \forall T_j\in\mathcal{T}^{S^*}\bigr\}$. To ensure consistency in distance scaling, a branch adjustment rate, denoted $\tau(T_i, \mathcal{T}^{S^*})$, is defined as the ratio of the sum of pairwise distances between common leaves in the target tree $T_i$ to the corresponding mean pairwise distances in $\mathcal{T}^{S^*}$ (Equation~\ref{eq:adj_rate}).

\begin{equation}
\tau(T_i, \mathcal{T}^{S^*}) = \frac{\sum\limits_{\{l_a,l_b\}\subseteq CL(T_i,\mathcal{T}^{S^*})} d^{(T_i)}(l_a, l_b)}{\sum\limits_{\{l_a,l_b\}\subseteq CL(T_i,\mathcal{T}^{S^*})} \bar{d}^{(\mathcal{T}^{S^*})}(l_a, l_b)}, \label{eq:adj_rate}
\end{equation}

where $\bar{d}^{(\mathcal{T}^{S^*})}(l_a, l_b)$ denotes the mean distance between leaves $l_a$ and $l_b$ across the trees in $\mathcal{T}^{S^*}$.

All branch lengths in $S^*$, including the connecting branch, are scaled by this factor (Equation~\ref{eq:new_length}).
\begin{equation}
\forall e \in E(S^*), \; |e| \gets \tau(T_i, \mathcal{T}^{S^*}) \cdot |e|. \label{eq:new_length}
\end{equation}

To determine the optimal insertion point $v^*$, the distances from each common leaf $l_c \in CL(T_i, \mathcal{T}^{S^*})$ to a candidate point $v$ in the intermediate tree $T'_i$ are evaluated and compared with a target distance derived from the contributing trees. For the consensus subtree $S^*$, let $att(S^*)$ be an attachment placeholder induced by the contributing source subtrees. For each $T_j \in \mathcal{T}^{S^*}$, the corresponding attachment node is $att(S_j)$, where $S_j \subseteq T_j$ is the rooted subtree at the LCA of $L(S^*)$ in $T_j$ and satisfies $L(S_j)=L(S^*)$.

For each $l_c\in CL(T_i,\mathcal{T}^{S^*})$, define the leaf-based adjustment rate
$\tau^{(l_c)}(T_i, \mathcal{T}^{S^*})$ by Equation~\ref{eq:leaf_adj_rate}.
\begin{equation}
\tau^{(l_c)}(T_i, \mathcal{T}^{S^*}) =
\frac{\sum\limits_{l_a \in CL(T_i,\mathcal{T}^{S^*})\setminus\{l_c\}} d^{(T_i)}(l_c, l_a)}
{\sum\limits_{l_a \in CL(T_i,\mathcal{T}^{S^*})\setminus\{l_c\}} \bar{d}^{(\mathcal{T}^{S^*})}(l_c, l_a)}.
\label{eq:leaf_adj_rate}
\end{equation}

The target distance is then defined by Equation~\ref{eq:temp_dist}.
\begin{equation}
d^{(\mathcal{T}^{S^*})}(l_c, S^*) =
\bar{d}^{(\mathcal{T}^{S^*})}(l_c, att(S^*)) \cdot \tau^{(l_c)}(T_i, \mathcal{T}^{S^*}),
\label{eq:temp_dist}
\end{equation}
where $\bar{d}^{(\mathcal{T}^{S^*})}(l_c, att(S^*))$ denotes $\frac{1}{|\mathcal{T}^{S^*}|}\sum_{T_j\in\mathcal{T}^{S^*}} d^{(T_j)}(l_c, att(S_j))$, and $S_j\subseteq T_j$ is the contributing subtree with $L(S_j)=L(S^*)$.

Candidate insertion positions are evaluated along each branch of the intermediate tree $T'_i$. However, only original branches $E(T_i)$ are considered. Even after prior insertions split an original branch $e=(u,w)$ into subsegments in $T'_i$, all such subsegments are treated as parts of the same original branch $e$ for candidate evaluation. Branches introduced by earlier insertions are never considered as candidate branches.

Fix an original branch $e=(u,w)\in E(T_i)$ of length $|e|$, where $u$ is the parent of $w$. For $x\in[0,1)$, let $v_e(x)$ be the point on $e$ at distance $x\cdot |e|$ from $u$ measured along $e$ toward $w$. The value $x=1$ is excluded to avoid representing the same node both as $x=1$ on one original branch and as $x=0$ on an adjacent original branch. If the minimizer on $e$ is attained at $w$, that insertion point is represented on an adjacent original branch incident to $w$ using $x=0$. Moving the insertion point $v_e(x)$ along the branch $e$ affects the leaves on the two sides of $e$ in opposite ways. For each common leaf $l_c\in CL(T_i,\mathcal{T}^{S^*})$, the distance to $v_e(x)$ can be written as $d^{(T'_i)}(l_c,u)+\varepsilon_{c,e}\,x\cdot |e|$, where $\varepsilon_{c,e}=+1$ if $l_c$ lies outside the subtree rooted at $w$ and $\varepsilon_{c,e}=-1$ if $l_c$ lies inside the subtree rooted at $w$.

The objective function is defined as the total squared deviation between the observed and target distances as follows (Equation~\ref{eq:obj_func}).
\begin{equation}
f_e(x) = \sum_{l_c \in CL(T_i, \mathcal{T}^{S^*})} \left( d^{(T'_i)}(l_c, u) + \varepsilon_{c,e} \cdot x \cdot |e| - d^{(\mathcal{T}^{S^*})}(l_c, S^*) \right)^2. \label{eq:obj_func}
\end{equation}

Since $f_e(x)$ is quadratic in $x$, for each fixed original branch $e\in E(T_i)$, its minimizer can be computed by
solving $\frac{d}{dx} f_e(x)=0$ to obtain the unconstrained minimizer and then projecting it to $[0,1)$, if necessary,
which results in $x_e^*\in[0,1)$.
The optimal insertion point is then $v^* = v_{e^*}(x_{e^*}^*)$, where $\displaystyle (e^*,x_{e^*}^*) \in \arg\min_{e\in E(T_i),x\in[0,1)} f_e(x)$. If multiple candidates attain the same minimal objective value, the one minimizing $d^{(T'_i)}(r_{T_i},v_e(x_e^*))$ is selected. If a tie remains, the candidate on the branch with the smallest fixed index
$\mathrm{idx}(e)$, given by a depth-first enumeration of the original branches $E(T_i)$, is selected. This results in a unique $v^*$.

Once $v^*$ is determined, $S^*$ is inserted at this position in $T'_i$, ensuring that discrepancies between observed and
target distances are minimized. The integration preserves the original topology of $T_i$ and maintains pairwise
relationships among the original taxa.

The insertion procedure is iterated over all consensus MCS until $L(T'_i) = L(T_i^{\uplus}) = L(\mathcal{T})$, which results in the inclusion of all missing taxa. The completed tree $T_i^{\uplus}$ is added to a newly constructed set of completed trees $\mathcal{T}^{\uplus}$ while leaving the original $T_i$ unchanged. The described process is repeated for each tree in $\mathcal{T}$, resulting in a fully completed phylogenetic tree set $\mathcal{T}^{\uplus}$.

\section{Properties}\label{sec:properties}

The following theorems and lemmas establish several properties of the proposed approach.

\begin{theorem}[Local optimality] \label{theorem:optimal}
Let $U_{t-1}$ be the set of taxa still uncovered before iteration $t$, let $\mathcal{S}_t$ be the candidate subtrees at that iteration, and let $p$ be the current threshold used in iteration $t$.
The selected subtrees $\mathcal{M}_t$ satisfy the following property.
For every $S \in \mathcal{S}_t$ whose leaf set satisfies $\bigl|\{\,T_j \in \mathcal{T}\setminus\{T_i\}\ \mid\ \exists\, S' \subseteq T_j,\ L(S')=L(S)\,\}\bigr|
\ge p \cdot \bigl|\mathcal{T}\setminus\{T_i\}\bigr|$, the inequality $\bigl|L(S)\cap U_{t-1}\bigr| \le \bigl|L(\mathcal{M}_t)\cap U_{t-1}\bigr|$ holds.
\end{theorem}

\begin{proof}
Let $T_i$ be a target tree, and let $U_0=L^{\uplus}(T_i)$ denote the initial set of uncovered taxa in $T_i$. At iteration $t$, let $U_{t-1}$ be the taxa still uncovered, and let $\mathcal{S}_t$ be the set of candidate subtrees. By construction, every subtree $S \in \mathcal{S}_t$ satisfies $L(S)\subseteq U_{t-1}$.

Subtrees in $\mathcal{S}_t$ are filtered by the current threshold $p$, and only subtrees in $\mathcal{S}'_t=\{\,S \in \mathcal{S}_t: P(L(S)) \ge p\,\}$ are retained. These subtrees are grouped into $\mathbb{G}_t = \bigl\{G\subseteq\mathcal{S}'_t : \forall S,\, S'\in G,\; L(S)=L(S')\bigr\}$, and a group $\mathcal{M}_t = \underset{G \in \mathbb{G}_t}{\mathrm{argmax}}\bigl|L(G)\bigr|$ is selected.

Let $S' \in \mathcal{S}'_t$ be any retained candidate subtree. Let $G'$ be the group such that $S' \in G'$. Then $L(S')=L(G')$, therefore $|L(S')| = |L(G')|\le|L(\mathcal{M}_t)|$.

Since $L(S')\subseteq U_{t-1}$ and $L(\mathcal{M}_t)\subseteq U_{t-1}$, this implies
$|L(S')\cap U_{t-1}| \le |L(\mathcal{M}_t)\cap U_{t-1}|$.
\end{proof}

\begin{remark} \label{remark:set-partition}
Theorem~\ref{theorem:optimal} establishes that a locally optimal choice is made in each iteration. A globally optimal solution, which would complete the tree using the absolute minimum number of insertions, may not be found by this greedy strategy. The underlying optimization problem is related to the \textsc{Exact Cover} problem, which is NP-complete \citep{karp2009reducibility}. The goal in set partitioning is to find a minimum-sized collection of \textit{disjoint} subsets whose union covers a universal set. Our algorithm serves as an efficient heuristic for this problem. We further analyze this step in Section~\ref{sec:param}, where we prove that it is fixed-parameter tractable with respect to a pair of parameters. This provides a way to compute optimal subtree selections for small instances, complementing the greedy heuristic.
\end{remark}
Let $n = |\mathcal{T}|$ denote the number of rooted binary phylogenetic trees in the set $\mathcal{T} = \{T_1, T_2, \dots, T_n\}$. Let $m = |L(\mathcal{T})|$ denote the total number of distinct taxa across the tree set. Each tree $T_i \in \mathcal{T}$ satisfies $|L(T_i)| \le m$, and therefore has at most $2m - 1$ nodes. Suppose further that every maximal completion subtree extracted from a source tree also has at most $m$ leaves.

\begin{lemma}[Selection of MCS complexity] \label{lemma:alg2_comp}
    The maximal completion subtree selection procedure runs in $O(n \cdot m^2)$ in the worst case.
\end{lemma}

\begin{proof}
Let $U$ be the set of uncovered taxa, initialized as the missing taxa $L^{\uplus}(T_i) = L(\mathcal{T}) \setminus L(T_i)$ for the target tree $T_i$. In the worst-case scenario, the selection of MCS incorporates only one taxon from $U$ in each iteration. Since $|U| \le m$, the number of iterations is at most $m$.

During each iteration, every source tree in $\mathcal{T} \setminus \{T_i\}$ is processed. Each source tree has at most $2m - 1$ nodes and $m$ leaves, and extracting candidate subtrees from a tree takes $O(m)$ time. Across all $n - 1$ source trees, this totals $O(n \cdot m)$ time per iteration.

To achieve $O(m)$ extraction time per source tree, \textsc{ExtractSubtrees} can be implemented by a single postorder traversal that, for every node, checks whether its descendant leaves are entirely contained in $U$, using $O(1)$ membership checks by storing $U$ as a hash set over $L(\mathcal{T})$. In addition, each candidate leaf set $L(S)$ can be stored as a bitset over $L(\mathcal{T})$ and looked up via hashing.

Let $\mathcal{S}_t$ be the set of candidate subtrees in iteration $t$. In the worst case, each of the $n - 1$ source trees contributes up to $m$ candidate subtrees, then $|\mathcal{S}_t| \le (n - 1) \cdot m$. For each candidate subtree $S$, the leaf set frequency $P(L(S))$ (Equation~\ref{eq:proportion_subtrees}) can be computed by counting, for each distinct leaf set, the number of source trees in which it occurs, using a hash table keyed by the hashed bitset representation of $L(S)$. For each fixed leaf set, the count is updated at most once per source tree. This takes $O(|\mathcal{S}_t|)$ time per iteration. The fallback mechanism for lowering the frequency threshold $p$ (in fixed steps, up to a constant number of times) does not affect the asymptotic complexity of the filtering step. Each fallback iteration reuses the same candidate set $\mathcal{S}_t$, and the total number of fallback checks is bounded by a constant. Therefore, the filtering step takes $O(|\mathcal{S}_t|)$ time per iteration.

Grouping subtrees by their leaf sets and selecting the group that maximizes coverage requires $O(|\mathcal{S}_t|)$ time and is dominated by the same order term.

Since each iteration takes $O(|\mathcal{S}_t|)=O(n \cdot m)$ time and there are at most $m$ iterations, the total worst-case time complexity is $O(n \cdot m^2)$.
\end{proof}

\begin{lemma}[Weighted majority-rule consensus complexity] \label{lemma:const_cons}
   The weighted majority-rule consensus tree construction runs in $O(n^2 \cdot m)$ time.
\end{lemma}
\begin{proof}
In the worst case, the set $\mathcal{M}_t$ contains at most $n - 1$ subtrees, one from each source tree in $\mathcal{T} \setminus \{T_i\}$.

A weight $w_j$ (Equation~\ref{eq:weights}) is computed for each source tree $T_j$. Since both $L(T_j)$ and $L(T_i)$ have at most $m$ taxa, computing their intersection takes $O(m)$ time. Across all source trees, this step takes $O(n \cdot m)$.

Each subtree $S_j \in \mathcal{M}_t$ contributes at most $m - 1$ splits, and $\mathcal{SP}$ contains at most $(n - 1) \cdot (m - 1)$ splits. Constructing the union $\mathcal{SP}$ requires $O(n \cdot m)$ time.

For each split $sp \in \mathcal{SP}$, computing its weighted frequency $f(sp)$ (Equation~\ref{eq:freq_split}) takes $O(n)$ time, as it involves checking up to $n - 1$ trees. Since $\mathcal{SP}$ contains at most $(n - 1) \cdot (m - 1)$ splits, this step takes $O(n^2 \cdot m)$ time. The set of splits with $f(sp)>0.5$ is compatible by the majority-rule property and therefore defines a rooted consensus topology on $L(\mathcal{M}_t)$.

The computation of branch lengths for majority-rule splits involves at most $m - 1$ splits, each requiring $O(n)$ time, resulting in a total of $O(n \cdot m)$. This does not exceed the previous step in asymptotic complexity.

The overall worst-case time complexity of the weighted majority-rule consensus tree construction is determined by its dominant term, $O(n^2 \cdot m)$.
\end{proof}

\begin{theorem}[Overall polynomial complexity] \label{theorem:alg_comp}
    The proposed phylogenetic tree set completion algorithm runs in $O(n^3 \cdot m^2 + n \cdot m^3)$ to complete all $n$ target trees.
\end{theorem}
\begin{proof}
Based on Lemma~\ref{lemma:alg2_comp}, the MCS selection step runs in $O(n \cdot m^2)$ time for a single target. Across $n$ target trees, this totals $O(n^2 \cdot m^2)$, which is dominated by the $O(n^3 \cdot m^2)$ term from the consensus construction.

At each iteration $t$ of the completion process, a set of maximal completion subtrees $\mathcal{M}_t$ is selected, and a consensus maximal completion subtree $S^*$ is constructed. Since there are at most $m$ iterations per target tree, and each iteration runs in $O(n^2 \cdot m)$ time (by Lemma~\ref{lemma:const_cons}), the total cost for this consensus construction step is $O(n^2 \cdot m^2)$ per target tree.

To insert $S^*$, a single minimizer is computed per original branch $e$. Since a binary tree with at most $m$ leaves has at most $2m - 2$ branches, the total number of candidate insertion positions is at most $2m - 2$. Each candidate is evaluated using Equation~\ref{eq:obj_func}, which sums over common leaves $l_c \in CL(T_i,\mathcal{T}^{S^*})$. In the worst case, the number of common leaves is at most $m$, thus evaluating one candidate takes $O(m)$ time. Therefore, evaluating all candidates takes $O(m^2)$ time.

Assuming at most $m$ insertions are needed per target tree, the integration step per tree costs $O(m^3)$. Across $n$ trees, the total integration cost is $O(n \cdot m^3)$.

Combining all components, the total worst-case time complexity is $O(n^3 \cdot m^2 + n \cdot m^3)$.
\end{proof}

\begin{theorem}[Correctness] \label{th:correctness}
For every $T_i\in \mathcal{T}$, the phylogenetic tree set completion algorithm outputs a unique completed tree $T_i^{\uplus}$ that satisfies the following: (i) $L(T_i^{\uplus}) = L(\mathcal{T})$, (ii) $\forall l_a, l_b \in L(T_i)$, $d^{(T_i^{\uplus})}(l_a, l_b)=d^{(T_i)}(l_a, l_b)$, and (iii) each inserted consensus subtree $S^*$ is placed at an attachment point $v^*$ that minimizes the quadratic objective function (Equation~\ref{eq:obj_func}) over all candidate points $v_e(x)$ with $e\in E(T_i)$ and $x\in[0,1)$.

\end{theorem}

\begin{proof}
Let $T_i \in \mathcal{T}$ be any target tree. The completed tree $T_i^{\uplus}$ is constructed by iteratively identifying and inserting consensus maximal completion subtrees $S^*$, as described in Section~\ref{sec:methods}. We prove each of the three required properties.

\textit{(i) Complete taxon coverage.} The uncovered taxon set is initialized as $U = L^{\uplus}(T_i) = L(\mathcal{T}) \setminus L(T_i)$, and at each iteration, a group of subtrees $\mathcal{M}_t$ is selected whose leaf set $L(\mathcal{M}_t) \subseteq U$ covers the maximum number of yet uncovered taxa.

Each iteration updates $U \leftarrow U \setminus L(\mathcal{M}_t)$. Since at least one new taxon is removed from $U$ at every iteration, and $|U| \le m$ initially, the loop must terminate in at most $m$ iterations.

Therefore, the union of inserted taxa satisfies $L(T_i^{\uplus}) = L(T_i) \cup \bigcup_t L(\mathcal{M}_t) = L(\mathcal{T})$, proving that all missing taxa are incorporated into $T_i^{\uplus}$.

\textit{(ii) Distance preservation between original taxa.} Throughout the insertion process, a working copy $T'_i$ of the target tree is maintained. Each consensus subtree $S^*$ is inserted into $T'_i$ along a branch originally present in $T_i$, ignoring branches and nodes introduced by earlier insertions. Thus, insertions occur only on original branches of $T_i$ and do not alter its internal topology.

Let an insertion place $S^*$ at $v^*=v_e(x^*)$ on an original branch $e=(u,w)$ of length $|e|$, with parameter $x^*\in[0,1)$ measured from $u$. The branch $e$ is replaced by two branches $(u,v^*)$ and $(v^*,w)$ with lengths $x^* \cdot |e|$ and $(1-x^*) \cdot |e|$, respectively. For any pair of original leaves $l_a,l_b\in L(T_i)$, the unique path in $T_i$ either does not traverse $e$, in which case it is unchanged, or traverses $e$ exactly once, in which case the segment of length $|e|$ is replaced by $x^* \cdot |e|+(1-x^*) \cdot |e|=|e|$. No other original branch is modified. Hence, $d^{(T'_i)}(l_a,l_b)$ is unchanged by this insertion, and $\forall l_a, l_b \in L(T_i), \; d^{(T_i^{\uplus})}(l_a,l_b)=d^{(T_i)}(l_a,l_b)$.

\textit{(iii) Optimal insertion point for $S^*$.} Fix a target tree $T_i$ and an inserted consensus subtree $S^*$. For each original branch $e\in E(T_i)$, the objective function $f_e(x)$ (Equation~\ref{eq:obj_func}) is quadratic in $x$,
and its minimizer $x_e^*\in[0,1)$ can be obtained by solving $\frac{d}{dx}f_e(x)=0$. Among all eligible original branches, the algorithm selects a branch
$\displaystyle e^*\in \arg\min_{e\in E(T_i)} f_e(x_e^*)$ and sets $v^*=v_{e^*}(x_{e^*}^*)$.
Therefore, each insertion is optimal with respect to minimizing deviation from the target distances induced by the
contributing source trees.

Together, the three components establish that the output $T_i^{\uplus}$ is (i) complete with respect to the full taxon set, (ii) distance preserving over the original leaf pairs, and (iii) constructed by optimal placement of each subtree.

Since all selection and insertion steps are deterministic (grouping by identical leaf sets, frequency thresholds, and tie-breaking), the tree $T_i^{\uplus}$ is unique for each $T_i$.
\end{proof}

\begin{theorem}[Order invariance] \label{th:unique}
The completed set $\mathcal{T}^{\uplus} =\{T_1^{\uplus},\dots,T_n^{\uplus}\}$ is independent of the order in which the target trees are processed.
\end{theorem}
\begin{proof}
Each $T_i$ is completed using only (i) its own original structure and (ii) the unmodified trees in $\mathcal{T} \setminus\{T_i\}$. Since $T_i^{\uplus}$ is added to a new set $\mathcal{T}^{\uplus}$ and the original set $\mathcal{T}$ remains unchanged, the data available for completing any $T_j$ is identical regardless of whether $T_j$ is processed before or after $T_i$.

According to the uniqueness established in Theorem~\ref{th:correctness}, each completed tree is therefore the same, independent of processing order. Moreover, the selection, consensus, and insertion operations (including tie-breaking rules) are deterministic, resulting in the same completed tree regardless of the processing order.
\end{proof}

\subsection{Parameterization} \label{sec:param}

The computational constraint in the phylogenetic tree set completion algorithm is associated with a globally optimal variant of the subtree selection step that minimizes the number of insertions. Specifically, given the set of missing taxa $U = L^{\uplus}(T_i)$ for a target tree $T_i$ and the family $\mathcal{F}$ of candidate leaf sets extracted from source trees, the objective of this variant is to select a subfamily of candidate subtrees whose union covers $U$ and whose sets are pairwise disjoint. This subsection formulates this optimal selection problem as a variant of the classic \textsc{Exact Cover} problem, proves NP-completeness, shows that the problem is fixed-parameter tractable when parameterized by $(k,\delta)$, where $k$ is an upper bound on the number of selected subsets and $\displaystyle\delta=\max_{F\in\mathcal{F}}|F|$ is the maximum candidate subset size. It also describes how the resulting FPT routine can be incorporated as an optional improvement together with a greedy fallback strategy.

Let $U = L^{\uplus}(T_i)$ be the set of taxa missing from the target tree $T_i$. Let $\mathcal{F}=\{F_1,\dots ,F_q\}$ be the family of leaf sets of all candidate subtrees extracted from the source trees, where $q = |\mathcal{F}|$. The objective is to select a subfamily $\mathcal{C}\subseteq\mathcal{F}$ of minimum size such that its sets are pairwise disjoint and their union is precisely $U$.

We analyze this problem through its decision variant, formalized as follows.

\begin{mdframed}[nobreak=true]
    \textsc{$(k,\delta)$-Exact Cover}\label{prob:exact-cover}\\
    \textbf{Input:} A universe $U$, a family of subsets $\mathcal{F}\subseteq 2^U$, an integer $k \in \mathbb{N}$, and $\displaystyle\delta=\max_{F\in\mathcal{F}}|F|$.\\
    \textbf{Question:} Does there exist a subfamily $\mathcal{C}\subseteq\mathcal{F}$ such that $|\mathcal{C}|\le k$, the sets in $\mathcal{C}$ are pairwise disjoint, and $\displaystyle \bigcup_{C\in\mathcal{C}} C = U$?
\end{mdframed}

\begin{lemma}\label{prop:nphard}
\textsc{$(k,\delta)$-Exact Cover} is NP-complete, even for $\delta=3$.
\end{lemma}

\begin{proof}
Exact Cover by 3-Sets (X3C) is a well-known NP-complete problem \citep{garey2002computers}. X3C is a specific instance of \textsc{$(k,\delta)$-Exact Cover}, where $\delta=3$ and we seek an exact cover with $k=|U|/3$ sets (which is the minimum possible size if all sets have size 3 and perfectly partition $U$). In an X3C instance, every exact cover of $U$ necessarily consists of $|U|/3$ disjoint 3-sets, therefore setting $\delta=3$ and $k=|U|/3$ yields an equivalent instance of \textsc{$(k,\delta)$-Exact Cover}.

Membership in NP is immediate. Given a candidate subfamily $\mathcal{C}$, we can check in polynomial time that $|\mathcal{C}|\le k$, that the sets in $\mathcal{C}$ are pairwise disjoint, and that $\displaystyle\bigcup_{C\in\mathcal{C}} C = U$.
\end{proof}

\begin{proposition}[Size-bound rejection]\label{prop:sizebound}
Let $(U,\mathcal{F},k)$ be an instance of \textsc{$(k,\delta)$-Exact Cover}. If $|U| > k\delta$, then the instance is a no-instance, i.e., no choice of at most $k$ pairwise disjoint sets from $\mathcal{F}$ can cover the entire universe $U$.
\end{proposition}

\begin{proof}
In any feasible solution, at most $k$ disjoint sets are selected, each containing at most $\delta$ elements, hence they can cover at most $k\delta$ distinct universe elements.
\end{proof}

Proposition~\ref{prop:sizebound} can be applied as a \textit{reduction rule} immediately after the family $\mathcal{F}$ has been obtained, i.e., $|U|>k\delta$ implies rejection. This check requires only the values of $|U|$ and $\delta$, both computable in a single linear scan.

\begin{theorem}[FPT in ($k, \delta$)]\label{thm:fpt}
\textsc{$(k,\delta)$-Exact Cover} is solvable in time $O(2^{k\delta} \cdot |\mathcal{F}| \cdot |U|)$, and is, therefore, fixed-parameter tractable when parameterized with respect to the
combined parameter $(k,\delta)$.
\end{theorem}

\begin{proof}
Let $U$ be the universe of missing taxa. By Proposition~\ref{prop:sizebound}, any instance where $|U| > k\delta$ can be immediately rejected. Thus, we assume $|U| \le k\delta$.

Let $\mathcal{F} \subseteq 2^U$ be the family of candidate leaf sets, with each $F \in \mathcal{F}$ satisfying $|F| \le \delta$.

A dynamic programming approach can be applied \citep{cygan2015parameterized}. We use dynamic programming over all subsets $X \subseteq U$ to store the minimum number of pairwise disjoint sets from $\mathcal{F}$ whose union is exactly $X$.

Define a table $\texttt{DP}[X]$, where $\texttt{DP}[X] = \min \left\{ \texttt{DP}[X \setminus F] + 1:\; F \in \mathcal{F}, F \subseteq X \right\}$, with the base case $\texttt{DP}[\emptyset] = 0$. For all $X \neq \emptyset$, we initialize $\texttt{DP}[X] = +\infty$ to indicate that $X$ is not yet known to be coverable. Each entry represents the smallest number of disjoint sets covering the subset $X$. Since at each transition, we only subtract sets $F \subseteq X$ from the current subset $X$, every element of $U$ is removed (covered) at most once in any sequence of transitions from $U$ to $\emptyset$. Hence, the sets used along such a sequence are pairwise disjoint, and their union is exactly the corresponding subset $X$.

The number of table entries is $2^{|U|} \le 2^{k\delta}$. For each subset $X$ and each $F \in \mathcal{F}$, we check whether $F \subseteq X$ and compute the update. These checks and updates take $O(|U|)$ time per entry. Therefore, the total time is $O(2^{k\delta} \cdot |\mathcal{F}| \cdot |U|)$.

Finally, we check whether $\texttt{DP}[U] \le k$. If this condition holds, there exists a collection of at most $k$ pairwise disjoint sets from $\mathcal{F}$ that exactly covers $U$.

Since the algorithm depends only exponentially on $k$ and $\delta$, the problem is fixed-parameter tractable for the combined parameter $(k, \delta)$.
\end{proof}

The FPT routine can be incorporated into the algorithm for the subtree selection step. At the start of each iteration, we compute $\displaystyle\delta =\max_{F\in\mathcal{F}}|F|$ (maximum size of a candidate subtree) and $k_{\min} =\Bigl\lceil\frac{|U|}{\delta}\Bigr\rceil$ (lower bound on the number of disjoint sets needed to cover $U$). If $U \neq \emptyset$, then $k_{\min}\ge 1$. Let $k_{\max}\in\mathbb{N}$ be a user-chosen upper bound with $k_{\max}\ge k_{\min}$. If the computed parameters satisfy user-chosen thresholds, an FPT routine is applied for increasing values of $k$ from $k_{\min}$ up to $k_{\max}$. For the first value of $k$ in this range for which the FPT routine finds a solution, the resulting subtree set is used, and the greedy step is bypassed. If the parameters exceed the thresholds, or if no solution is found for any $k \in \{k_{\min},\dots,k_{\max}\}$, the greedy rule ($\mathrm{argmax}|L(G)|$) described in Section~\ref{sec:methods} is applied. This strategy ensures that the algorithm always runs in polynomial time on the full instance while providing optimal solutions for instances with small parameter values.

\subsection{Resolving multifurcations} \label{sec:res-multifurc}

The phylogenetic tree set completion algorithm may produce internal nodes with a degree greater than two, meaning that completed trees may include multifurcations. To make the output compatible with phylogenetic methods that assume bifurcating trees, an optional deterministic multifurcation resolution can be applied. This procedure is a refinement of the completed tree in the sense that it resolves multifurcations by inserting only zero-length internal branches, thereby preserving the leaf set and all pairwise leaf-to-leaf distances.

\begin{proposition}[Multifurcation resolution] \label{prop:multifurc-resolution}
    Let $T_i^{\uplus}$ be the completed tree produced for a target tree $T_i$ by the phylogenetic tree set completion algorithm, and assume that all leaf labels are unique. Consider the following multifurcation resolution procedure. For every internal node $v$ of $T_i^{\uplus}$ with more than two children, repeatedly order the children of $v$ by the lexicographically smallest leaf label among their descendant leaves, select the first two children in this order, and replace them by a new internal node $w$ attached to $v$ with a branch length of zero such that the selected children become descendants of $w$ with their original branch lengths preserved. Let $\widetilde{T}_i^{\uplus}$ denote the tree obtained after this procedure has been applied to all nodes of $T_i^{\uplus}$. Then, $\widetilde{T}_i^{\uplus}$ is a rooted phylogenetic tree in which every internal node has at most two children, and it satisfies the following properties. (i) For all leaves $l_a, l_b \in L(\mathcal{T})$, the pairwise distances satisfy $d^{(\widetilde{T}_i^{\uplus})}(l_a, l_b) = d^{(T_i^{\uplus})}(l_a, l_b)$. (ii) $L(\widetilde{T}_i^{\uplus}) = L(T_i^{\uplus}) = L(\mathcal{T})$. (iii) This refinement is idempotent, and $\widetilde{T}_i^{\uplus}$ is a refinement of $T_i^{\uplus}$ obtained only by resolving multifurcations via zero-length internal branches.
\end{proposition}

\begin{proof}
The multifurcation resolution procedure operates locally on each internal node $v$ of $T_i^{\uplus}$. Suppose $v$ has children $c_1, \dots, c_{\deg}$ with $\deg > 2$. A single refinement step selects two children deterministically according to the lexicographic order of the smallest leaf label among the descendant leaves of each child subtree and introduces a new internal node $w$ as a child of $v$ with a branch length of zero. The chosen children, $c_1$ and $c_2$, are detached from $v$ and reattached as children of $w$ with the same branch lengths as before.

After this transformation, the path length from $v$ to each of $c_1$ and $c_2$ is $d(v, w) + d(w, c_j) = 0 + d(v, c_j) = d(v, c_j), \; j \in \{1,2\}$, and the path lengths from $v$ to all other children of $v$ remain unchanged. Hence, the distance from the root of $T_i^{\uplus}$ to every leaf is preserved by this operation. A sequence of such multifurcation resolution steps at nodes with more than two children, therefore, preserves all root to leaf distances in $T_i^{\uplus}$.

In order to show that all pairwise leaf distances are preserved, fix any two leaves $l_a$ and $l_b$ and consider their unique path in $T_i^{\uplus}$. Each refinement step only modifies the local structure at some node $v$ by inserting a new node $w$ on a zero-length branch and reattaching two child subtrees beneath $w$ without changing any existing branch lengths. Any segment of the path between $l_a$ and $l_b$ that does not traverse $v$ is unchanged. If the path traverses $v$ and reaches one of the reattached child subtrees, then the segment $v \to c_j$ is replaced by $v \to w \to c_j$ with $d(v,w)=0$ and $d(w,c_j)=d(v,c_j)$. Therefore, the total path length is unchanged. Since each refinement step preserves the length of the path between $l_a$ and $l_b$, the refined tree satisfies $d^{(\widetilde{T}_i^{\uplus})}(l_a, l_b) = d^{(T_i^{\uplus})}(l_a, l_b)$ for all leaf pairs, thereby establishing the statement (i).

The refinement only introduces internal nodes and never deletes or relabels leaves, therefore, the leaf set is preserved, supporting the statement (ii). The multifurcation resolution continues at a node $v$ until $v$ has at most two children. After processing all nodes of $T_i^{\uplus}$ in this manner, every internal node in $\widetilde{T}_i^{\uplus}$ has at most two children. The resulting tree $\widetilde{T}_i^{\uplus}$ is thus rooted and satisfies $L(\widetilde{T}_i^{\uplus}) = L(\mathcal{T})$.

Idempotence follows from the fact that the refinement condition at a node requires strictly more than two children. Once the procedure terminates, every node in $\widetilde{T}_i^{\uplus}$ has at most two children, and thus no further refinement steps are applicable. Moreover, each refinement step modifies the tree only by inserting a zero-length internal branch and reattaching two existing child subtrees beneath the inserted node while keeping all original branch lengths. Thus, $\widetilde{T}_i^{\uplus}$ is a refinement of $T_i^{\uplus}$ obtained only by resolving multifurcations. Finally, the deterministic child pair selection rule ensures that at every step the same pair is chosen given the same leaf labels and tree structure, therefore, the refined tree $\widetilde{T}_i^{\uplus}$ is uniquely determined by $T_i^{\uplus}$.
\end{proof}

Theorem~\ref{th:correctness} and Theorem~\ref{th:unique} imply that each target tree $T_i$ has a unique completed tree $T_i^{\uplus}$. Proposition~\ref{prop:multifurc-resolution} shows that the optional multifurcation resolution step produces a uniquely defined tree $\widetilde{T}_i^{\uplus}$ with no multifurcations, the same leaf set, and identical pairwise distances. The overall completion procedure, therefore, remains deterministic both in the multifurcating and in the refined variants.

\section{Experimental evaluation}\label{sec:experiments}

\subsection{Datasets}

Four species groups, including amphibians, mammals, sharks, and squamates, are utilized in the evaluation. The biological data used to construct overlapping tree sets are obtained from the VertLife resource \citep{upham2019inferring}, using published comprehensive phylogenetic trees of amphibians \citep{jetz2018interplay}, mammals \citep{upham2019inferring}, sharks \citep{stein2018global}, and squamates \citep{tonini2016fully}. For each species group, a list of all species is compiled from the available taxonomy, and a random sampling procedure is applied to generate 20 base taxon subsets with increasing numbers of species, ranging from 50 to 145 in increments of 5. This step produces variation in subset size and taxon composition, which increases diversity in the resulting trees. Each base taxon subset is then used to prune the broad-coverage VertLife phylogenetic tree for the given group to produce a corresponding reference tree while preserving branch lengths.

From each reference tree, an overlapping tree set is created, consisting of 30 trees obtained by pruning to partially intersecting taxa sets derived from the corresponding base taxon subset. The taxon subsets are generated under explicit overlap constraints, using Jaccard overlap in the range 0.30--0.70 and enforcing full coverage of the base taxa subset across the set. Each subset additionally includes a fixed set of 10 anchor taxa to provide a stable shared basis across trees. The pruning procedure removes leaves but does not modify topologies or branch lengths during pruning.

In order to incorporate controlled noise among input trees, additional modifications are applied after pruning. Topological variation is introduced by performing three nearest-neighbor interchange (NNI) moves per tree, while protecting anchor taxa from being moved. Branch lengths are modified by multiplying all branches by a global scaling factor sampled uniformly from 1.003 to 1.009.

The union of taxa within each overlapping tree set defines the completion target for subsequent experiments, and the unpruned reference tree serves as the ground truth for that target. This design establishes a known target tree for each dataset to evaluate completion under taxon incompleteness and controlled discordance in both topology and branch lengths across the input trees. Consequently, any discrepancy between a completed tree and its reference arises from how the completion method combines overlapping, taxonomically incomplete subtrees with minor topological and branch-length modifications. All prepared datasets with more detailed descriptions and scripts used in this evaluation are openly available in the project GitHub repository.

\subsection{Baseline methods} 
In order to evaluate performance and to isolate the contribution of individual design choices (maximal subtree grouping, distance-based attachment, and the search over continuous attachment points), four baseline completion strategies are compared to the proposed phylogenetic tree set completion algorithm.

(i) \textit{Root Attach} method (naive baseline). Each consensus MCS is attached directly at the root of the target tree. This baseline tests whether improvements are due to non-trivial placement rather than to subtree identification alone.

(ii) \textit{Nearest Parent} method (discrete placement baseline). Each consensus MCS is attached to the parent of a selected common leaf in the target tree, where the selected leaf is the one that minimizes the objective function. This baseline is more informative than the Root Attach strategy because it places subtrees at variable locations that depend on the target tree (rather than always attaching at the same fixed point). This method tests whether a simple local placement rule already performs well. Compared to the proposed method, it restricts attachment points to a discrete set (parents of common leaves) instead of optimizing over all points along all original branches.

(iii) \textit{No MCS} method (single leaf insertion baseline). Missing taxa are inserted one by one by treating each missing leaf as its own completion subtree and placing it using the same objective function. This isolates the effect of grouping missing taxa into larger maximal completion subtrees and inserting them jointly.

(iv) \textit{Supertree-based} method (global majority-rule supertree baseline). A majority-rule supertree $T^*$ on the full taxon set $L(\mathcal{T})$ is constructed once from the entire input tree set and is then used as a single source tree to complete all target trees. In this evaluation, $T^*$ is constructed as a majority-rule (+) supertree using PluMiST MR(+) \citep{kupczok2011split}, and the resulting $T^*$ is fixed for the corresponding tree set. For a target tree $T_i$, completion is performed by pairwise tree completion with $T^*$, inserting the missing taxa $L(\mathcal{T}) \setminus L(T_i)$ into $T_i$ using only information from $T^*$. Branch lengths on $T^*$ are assigned by split-wise averaging over the input trees that contain the corresponding split.

Baselines (i)--(iii) are controlled ablations. Methods (i)--(ii) use the same extracted maximal completion subtrees and the same consensus construction as the proposed method but replace the subtree insertion rule. Method (iii) disables MCS grouping but uses the same insertion strategy as the proposed method. This makes it possible to attribute performance differences to the insertion strategy or the distinct leaf grouping step. The supertree-based method is the strongest baseline since it uses more comprehensive information about trees in a tree set.

\subsection{Distance measures}

Two distance measures, one topology-based and one branch length-based, are computed for each tree in the completed tree sets against its reference tree. In each such comparison, the two trees are defined on the same leaf set. Denote this leaf set by $L$. The Robinson-Foulds (RF) distance \citep{robinson1981comparison} measures topological dissimilarity by comparing the sets of non-trivial clades induced by internal nodes. The normalized RF distance is calculated by dividing RF by its maximum possible value. This distance can become less informative if trees contain multifurcations, because unresolved nodes induce fewer non-trivial clades and can artificially reduce RF. For this reason, the refined variant of the tree completion output (Proposition~\ref{prop:multifurc-resolution}) is used for distance evaluation, where multifurcations are deterministically resolved by inserting only zero-length internal branches.

Branch Score Distance (BSD) \citep{steel1993distributions,koshkarov2024novel} quantifies differences between two trees in terms of their pairwise leaf-to-leaf distances. Fix a consistent ordering of unordered leaf pairs $\{\{l_a,l_b\}: l_a,l_b\in L,\ a<b\}$ (e.g., induced by the lexicographic order of taxon labels). Let $\mathbf{d}_1$ be the vector whose entries are the distances $d^{(T_1)}(l_a,l_b)$ in $T_1$ in that order, and let $\mathbf{d}_2$ be defined analogously for $T_2$ using the same ordering of leaf pairs. Then, $\mathrm{BSD}(T_1,T_2)=\lVert \mathbf{d}_1 - \mathbf{d}_2\rVert_2$.

To make values comparable across subsets with different overall distance scales, BSD is normalized using the Euclidean norm of the element-wise maximum of the two distance vectors. The element-wise maximum is the vector whose $i$-th entry equals $\max\{(\mathbf{d}_1)_i,(\mathbf{d}_2)_i\}$. The normalized BSD is defined as follows (Equation~\ref{eq:bsd_norm}).

\begin{equation} \label{eq:bsd_norm}
\mathrm{BSD}_\text{norm}(T_1,T_2)=
\frac{\lVert \mathbf{d}_1 - \mathbf{d}_2\rVert_2}
{\lVert \max(\mathbf{d}_1,\mathbf{d}_2)\rVert_2}.
\end{equation}

Since pairwise leaf-to-leaf distances are strictly positive, $\lVert \max(\mathbf{d}_1,\mathbf{d}_2)\rVert_2>0$, the normalization in Equation~\ref{eq:bsd_norm} is well-defined. This normalization is bounded between 0 and 1 because the entry-wise absolute difference cannot exceed the entry-wise maximum. It is scale-invariant since multiplying all branch lengths in both trees by a common constant scales the numerator and denominator by the same factor. Since Proposition~\ref{prop:multifurc-resolution} preserves all pairwise leaf-to-leaf distances, BSD and $\mathrm{BSD}_\text{norm}$ remain unaffected by multifurcation resolution.

For every subset, all methods complete the same input sets, and each completed tree is compared to the subset reference tree to obtain both distances. To aggregate results within each subset, the median over completed trees is computed for each method and each distance. Using the median provides a robust summary across tree sets and gives one value per method per subset, which prevents subsets with larger taxa sets from dominating the analysis.

\subsection{Evaluation results}

\subsubsection{Comparative evaluation of completion methods}

Figures~\ref{fig:rf_norm}--\ref{fig:bsd_norm} summarize completion accuracy across the four datasets using subset-level medians of the normalized RF distance and normalized BSD. Each panel corresponds to one species group, including (a) Amphibians, (b) Mammals, (c) Sharks, and (d) Squamates. Across all species groups and both metrics, the proposed phylogenetic tree set completion method achieves the best overall performance, with smaller median distances than all baselines. The supertree-based baseline is consistently the closest competitor, whereas the three ablation baselines (Root Attach, Nearest Parent, and No MCS) perform worse.

\begin{figure}[h]
    \centering
    \includegraphics[width=1\linewidth]{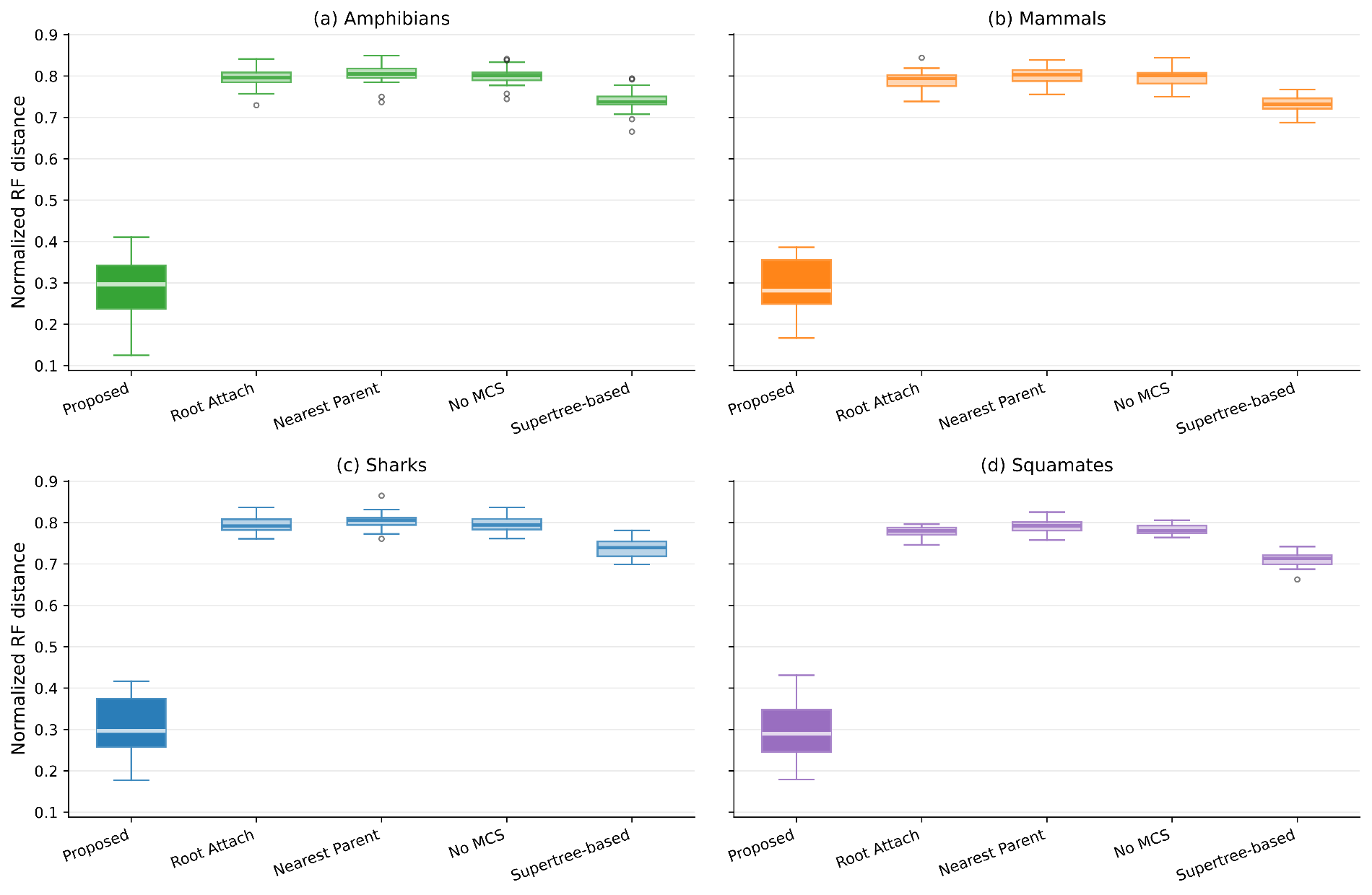}
    \caption{Normalized RF distance for completed trees relative to the subset reference tree. Each panel corresponds to one dataset, with (a) Amphibians, (b) Mammals, (c) Sharks, and (d) Squamates. For each subset and each method, the normalized RF distance is computed for every completed tree and then aggregated using the median across the 30 completed trees, resulting in one value per subset.  Boxplots summarize these 20 subset medians per method. Lower values indicate better topological concordance with the reference tree. The proposed method consistently outperforms all baseline methods across all species groups. This is also shown by a clear gap in the boxplots.}
    \label{fig:rf_norm}
\end{figure}

Based on normalized RF distances, shown in Figure~\ref{fig:rf_norm}, the proposed method substantially reduces topological dissimilarity relative to all baselines. Across the 80 subsets, the overall median normalized RF equals 0.290 for the proposed method, compared to 0.733 for the supertree-based baseline, meaning that the median normalized RF is approximately 60\% lower for the proposed method. The ablation baselines remain substantially worse, with overall medians approximately in the range 0.792 to 0.801.

The advantages of the proposed method are not limited to the central tendency. For every species group, the worst-case subset median normalized RF of the proposed method remains below the best case subset median normalized RF of the supertree-based baseline. Amphibians show a maximum of 0.411 for the proposed method versus a minimum of 0.665 for the supertree-based baseline. Mammals show 0.386 versus 0.688. Sharks show 0.417 versus 0.699. Squamates show 0.431 versus 0.663. This separation produces a clear visual gap in the boxplots and indicates that the proposed method dominates the baselines in topology at the subset level and attains the lowest normalized RF in all subsets.

Among the ablation methods, the Root Attach method is the best performing baseline on normalized RF, but it remains far from the proposed method, with median normalized RF values of 0.794 in mammals and 0.792 in sharks. The No MCS method is consistently slightly worse than the Root Attach method, with median normalized RF values ranging from 0.781 in squamates to 0.802 in mammals. The Nearest Parent method performs worst among the ablations on normalized RF, with median values around 0.792 to 0.805 across groups. Root Attach and Nearest Parent use the same extracted completion subtrees and the same consensus construction as the proposed method but replace the insertion rule. Their high normalized RF values indicate that subtree identification alone is insufficient and that the main topological gains come from optimizing the attachment location along the original branches. Disabling MCS grouping in No MCS increases normalized RF relative to Root Attach, which supports the benefit of inserting missing taxa jointly as maximal completion subtrees rather than as single leaves.

The normalized BSD results (Figure~\ref{fig:bsd_norm}) follow the same overall ranking while showing partial overlap between methods, which is expected given that BSD is sensitive to branch length scaling and local distortions. Overall, the proposed method attains the smallest normalized BSD, with a median of 0.257 over all subsets, compared to 0.317 for the supertree-based baseline, i.e., the median normalized BSD is about 19\% lower for the proposed method. The ablation baselines again remain worse, with overall medians approximately in the range 0.378 to 0.381.

\begin{figure}[h]
    \centering
    \includegraphics[width=1\linewidth]{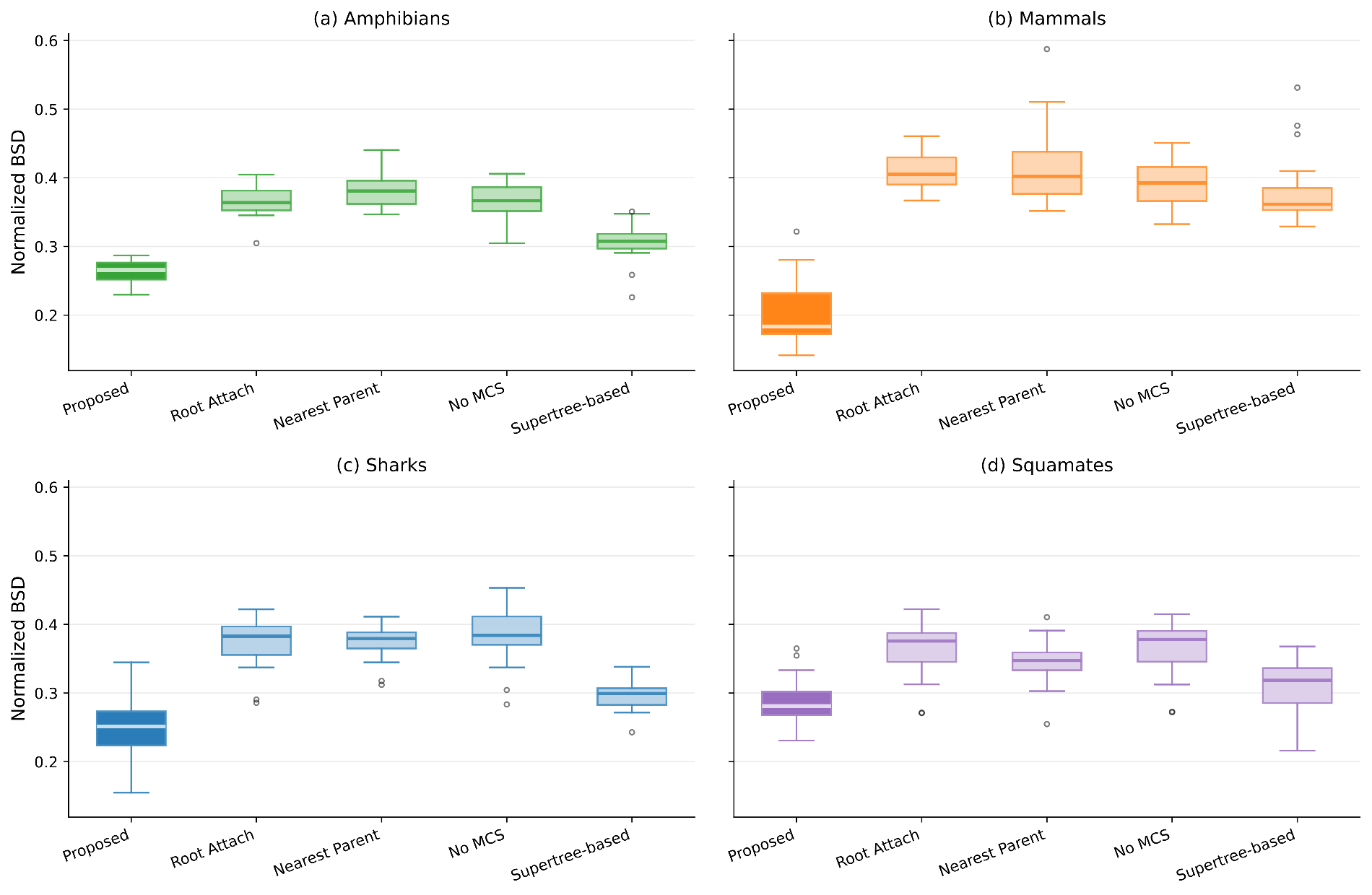}
    \caption{Normalized BSD for completed trees relative to the subset reference tree.  Each panel corresponds to one dataset, including  (a) Amphibians, (b) Mammals, (c) Sharks, and (d) Squamates. For each subset and each method, the normalized BSD is computed for the completed and reference trees, and then aggregated using the median across the 30 completed trees to obtain one value per subset. Boxplots summarize these 20 subset medians per method. Lower values indicate greater pairwise distance consistency between the completed and reference trees. Overall, the proposed method achieves the lowest normalized BSD.}
    \label{fig:bsd_norm}
\end{figure}

The supertree-based method remains the strongest baseline and the closest competitor. Across all 80 subsets, the proposed method has a median normalized BSD of 0.257, with an interquartile range of 0.226 to 0.281, while the supertree-based baseline has a median of 0.317, with an interquartile range of 0.296 to 0.348. The difference is largest in mammals, where the median normalized BSD declines from 0.362 to 0.183, and smallest in squamates, where the median declines from 0.318 to 0.281. Mammals also show particularly clear visual separation in BSD, with minimal overlap between the proposed method and the baselines (see Figure~\ref{fig:bsd_norm} (b)).

Ablation baselines remain worse than the proposed method on the normalized BSD. In mammals, the best ablation baseline, No MCS at 0.393, remains far above the proposed method at 0.183. This supports the fact that the attachment strategy combined with MCS grouping improves not only the tree topology but also the pairwise path-length structure, measured by the normalized BSD.

Consistency across subsets is additionally assessed using paired comparisons since each subset is completed by every method. A one-sided Wilcoxon signed-rank test is applied across the 20 subsets in each group, using aligned subset indices, to test the alternative hypothesis that the proposed method achieves smaller distances than each baseline method. Statistical significance is assessed at $\alpha=0.05$, with $p$-values adjusted using the Holm correction to account for the four baseline comparisons per group and metric. Several comparisons yield identical Holm-corrected $p$-values because the proposed method attains smaller subset medians in all 20 paired subsets, which gives the minimum attainable exact one-sided Wilcoxon $p$-value $2^{-20}$ and $4\cdot 2^{-20}$ after Holm correction across four baselines.

In particular, for normalized RF, the proposed method is significantly better than every baseline in all groups (Holm-corrected $p\text{-value}=3.82\times 10^{-6}$ for each comparison), consistent with the proposed method achieving the lowest subset median normalized RF in every subset. For normalized BSD, the proposed method is significantly better than the supertree-based baseline in all four groups (Holm-corrected $p\text{-value}=3.82\times 10^{-6}$ for amphibians, $3.82\times 10^{-6}$ for mammals, $1.81\times 10^{-5}$ for sharks, and $3.79\times 10^{-2}$ for squamates). Relative to the ablation baselines, the proposed method also shows significantly smaller normalized BSD in all groups (Holm-corrected $p\text{-value}=3.82\times 10^{-6}$ for each comparison).

Overall, the proposed phylogenetic tree set completion method produces completed tree sets that are closer to the reference trees than all baseline methods, with particularly large gains in topology, as measured by the normalized RF. The supertree-based approach is consistently the second best method, which suggests that aggregating global phylogenetic signals is useful, but using a single fixed supertree to complete all targets cannot adapt well to individual target trees and subsets. The ablation methods show that attachment strategy is crucial since simple attachment rules lead to near maximal topological dissimilarity, and that grouping missing taxa into maximal completion subtrees is beneficial, since disabling it increases the normalized RF. Together, these results support combining MCS grouping with objective-driven attachment to improve completion in both topology and branch lengths.

\subsubsection{Subset size trends analysis}

This analysis examines how completion accuracy varies as subset size increases from 50 to 145 taxa. Figure~\ref{fig:trend_subset_size} reports trends separately for each species group. For a fixed group and method, each line point is the subset-level median distance over the 30 completed target trees in that subset, and the shaded band shows the interquartile range (IQR) of the corresponding tree distances within the same subset.

For the median normalized RF (Figure~\ref{fig:trend_subset_size}(a)), the proposed method shows an overall increase with subset size, with non-monotonic fluctuations. In addition to the group-specific lines shown in Figure~\ref{fig:trend_subset_size}, a cross-group summary is computed by taking, at each subset size, the median of the four group-specific subset medians for the proposed method. Under this cross-group summary, the median increases from about 0.172 at 50 taxa to about 0.388 at 145 taxa, and it attains a local minimum near 80 taxa at about 0.199.

\begin{figure}[!t]
    \centering
    \includegraphics[width=1\linewidth]{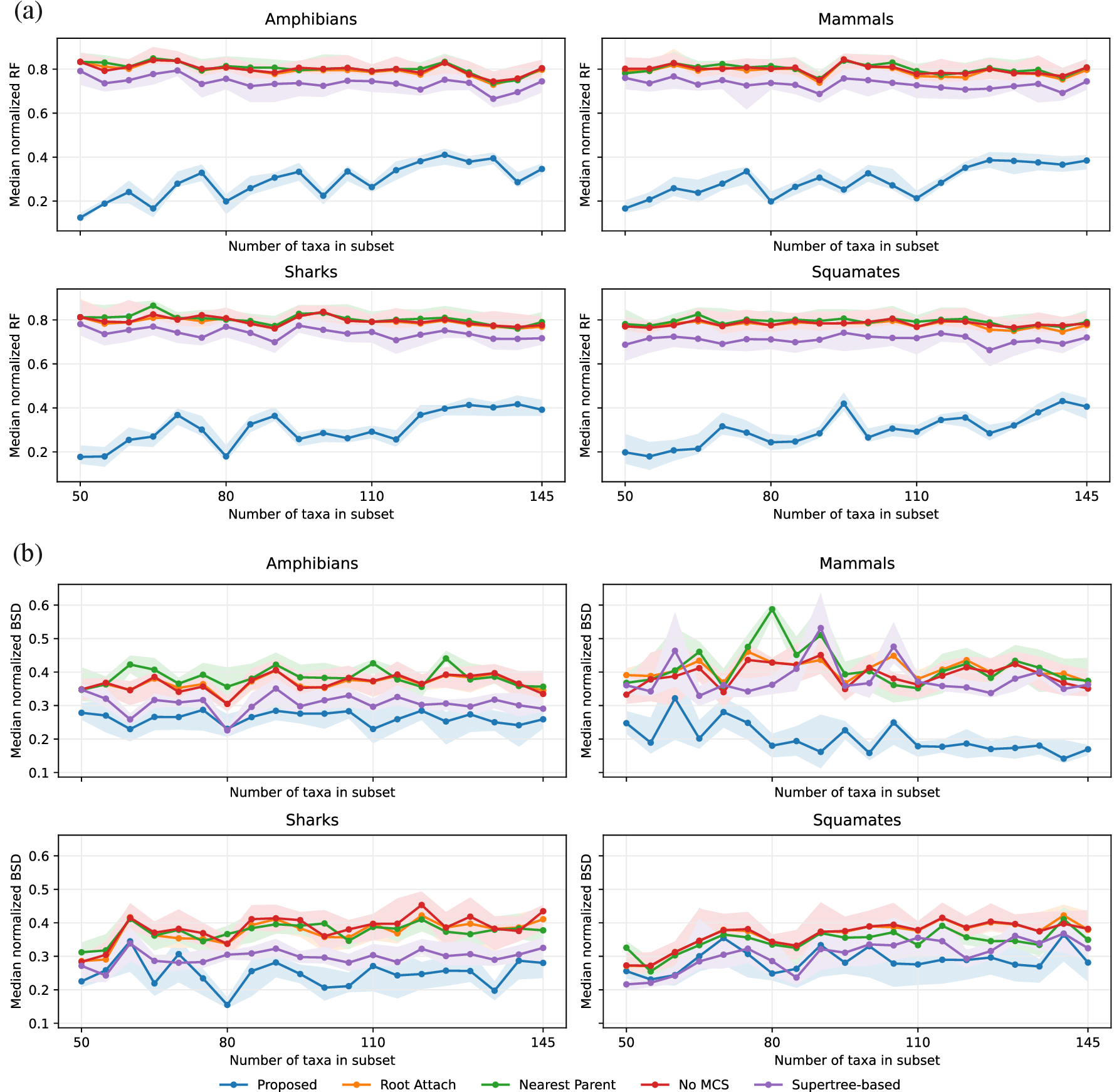}
    \caption{Subset size trends for completion accuracy across the four species groups. Panel (a) shows median normalized RF and panel (b) shows median normalized BSD. Within each panel, the four subplots correspond to Amphibians, Mammals, Sharks, and Squamates. For each method and each subset size, the line shows the subset-level median distance computed over the 30 completed target trees in that subset. Lower values indicate closer fit to the subset reference tree. Shaded bands show the interquartile range, defined as the 25th and 75th percentiles of the corresponding distances over the same 30 completed trees. Subset sizes range from 50 to 145 taxa in increments of 5.}
    \label{fig:trend_subset_size}
\end{figure}

For the median normalized BSD (Figure~\ref{fig:trend_subset_size}(b)), the proposed method is comparatively stable under the same cross-group summary. The median changes only slightly from about 0.252 at 50 taxa to about 0.270 at 145 taxa, with a local minimum near 80 taxa at about 0.205. At the level of individual species groups, amphibians and mammals tend to show decreasing BSD with increasing subset size, whereas sharks and squamates tend to show increasing BSD. Since these trends move in opposite directions across groups, the cross-group summary changes only modestly with subset size because increases in some groups are offset by decreases in others.

The IQR bands further characterize within-subset variability across the 30 completed target trees. For the median normalized RF under the proposed method, the median IQR width across groups is essentially unchanged between the smallest and largest subsets, about 0.070 at 50 taxa versus about 0.069 at 145 taxa, indicating broadly consistent within-subset dispersion across the size range. For the median normalized BSD, within-subset variability is more heterogeneous across sizes and groups, and the median IQR width increases modestly from about 0.050 at 50 taxa to about 0.078 at 145 taxa.

Overall, the proposed method exhibits stable behavior for median normalized BSD across subset sizes under the cross-group summary and a gradual increase in median normalized RF as subset size grows, while maintaining consistent within-subset variability for median normalized RF and moderate variability for median normalized BSD.

\subsubsection{Majority-rule consensus evaluation}

This analysis assesses how well each method preserves the global phylogenetic signal of its completed tree set, which in this context refers to both the dominant topological structure and the associated branch lengths. Preservation is quantified by the distance between a single representative tree (consensus or supertree) and the subset reference tree under normalized RF and normalized BSD.

For every subset and method, all completed trees produced for that subset are used to construct a rooted majority-rule consensus tree. Branch lengths are assigned on the consensus tree by split-wise averaging over the completed trees, using the mean length associated with each retained split among the trees that contain it. The resulting consensus tree is then compared to the subset reference tree using both normalized RF distance and normalized BSD. In parallel, the majority-rule supertree (MR(+)) used in the supertree-based baseline (one per tree set) is compared to the same reference tree under the same two metrics. Results are summarized as heatmaps over subsets (Figure~\ref{fig:heatmaps}).

\begin{figure}[h]
    \centering
    \includegraphics[width=1\linewidth]{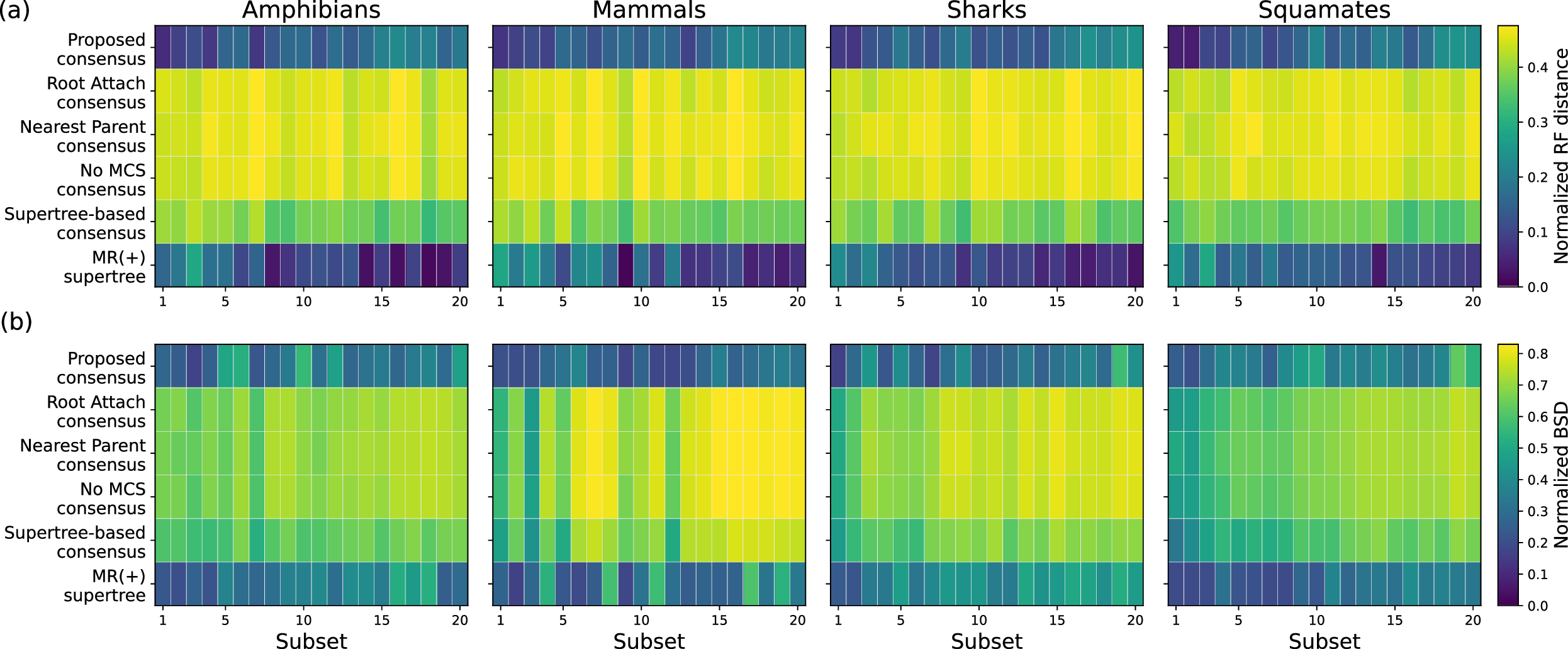}
    \caption{Heatmaps of distances across subsets for each species group. Columns correspond to Amphibians, Mammals, Sharks, and Squamates. Part (a) shows normalized RF distances and part (b) shows normalized BSD, each computed between the subset reference tree and a majority-rule consensus tree for the corresponding method (proposed method consensus and four baseline method consensuses), and the MR(+) supertree. Lower values (darker hues) indicate closer correspondence to the reference. Overall, the MR(+) supertree and the proposed method consensus attain the smallest distances.}
    \label{fig:heatmaps}
\end{figure}

Across all subsets, the two lowest-distance approaches are the MR(+) supertrees and the consensus derived from trees completed by the proposed method. MR(+) supertrees tend to produce smaller normalized RF distances, whereas the consensus from the proposed method tends to achieve smaller normalized BSD.

For normalized RF (Figure~\ref{fig:heatmaps}(a)), the MR(+) supertrees have a median of 0.117 with an IQR of 0.075--0.160, and the proposed method consensus has a median of 0.153 with an IQR of 0.119--0.192. For normalized BSD (Figure~\ref{fig:heatmaps}(b)), the proposed method consensus has a median of 0.310 with an IQR of 0.253--0.377, and the MR(+) supertrees have a median of 0.327 with an IQR of 0.257--0.406. All remaining consensus constructions show noticeably larger distances to the reference. In particular, the consensus built from the supertree-based completion has a median normalized RF of 0.373 and a median normalized BSD of 0.644. The other baselines demonstrate higher values, with median normalized RF between 0.452 and 0.457 and median normalized BSD between 0.713 and 0.720.

The proposed method consensus and MR(+) supertrees are therefore close overall, with a difference between their medians of 0.035 in normalized RF (favoring MR(+)) and 0.016 in normalized BSD (favoring the proposed consensus). Examining these data subset by subset, the median absolute gap is 0.094 for normalized RF and 0.091 for normalized BSD. Approximately 75\% of subsets have an absolute gap of at most 0.144 (RF) and 0.149 (BSD). Considering both metrics jointly, the proposed consensus is within 0.15 of the MR(+) supertree in 61\% of subsets, and within 0.20 in 80\% of subsets.

This pattern is also observed across species groups. For normalized RF, the proposed method consensus medians range from 0.144 to 0.160 across groups, compared to 0.095 to 0.141 for MR(+), indicating consistently small but generally lower RF distances for MR(+). For normalized BSD, the proposed method consensus attains low medians in amphibians (0.311), mammals (0.263), and sharks (0.297), compared to 0.325, 0.273, and 0.423 for MR(+), while in squamates MR(+) is smaller (0.293 versus 0.366).

Based on these results, the proposed method outputs completed tree sets with small topological and branch-length discrepancies. As an additional evaluation step, a single representative consensus tree is constructed from each completed set. The resulting consensus trees achieve low normalized RF and normalized BSD and remain close to the MR(+) supertree, indicating that completion preserves the global phylogenetic signal, as quantified by small distances to the reference, required for subsequent analyses.

\section{Conclusion}

In this work, we introduced the phylogenetic tree set completion problem and developed a deterministic polynomial-time algorithm for completing overlapping collections of phylogenetic trees. Our approach iteratively selects, constructs, and inserts consensus maximal completion subtrees while preserving the original topology of the target trees and all pairwise distances among the original leaves. This results in a completed set of phylogenetic trees on the full taxon union. Our new approach addresses an important gap in phylogenetic analysis and provides a foundation for integrating incomplete evolutionary information with the preservation of key phylogenetic signals.

The proposed method integrates phylogenetic signal (topology and branch lengths) from overlapping source trees through an overlap-weighted consensus construction and a distance-based objective function for choosing attachment locations. The procedure runs in polynomial time, produces a unique completion for each target tree, and is order-independent with respect to the processing order of target trees. By formulating the optimal coverage step as an exact-cover problem, we established NP-completeness and proved fixed-parameter tractability for the combined parameter $(k,\delta)$. We additionally presented an optional deterministic multifurcation resolution procedure that refines completed trees by inserting only zero-length internal branches. This refinement preserves all leaf-to-leaf distances and maintains the determinism of the overall workflow.

Our empirical evaluation on amphibians, mammals, sharks, and squamates shows that, among completion methods that preserve the target subtree on $L(T_i)$, the completed trees produced by our method are consistently closer to the corresponding subset reference trees than those produced by all baselines, both in topology (normalized RF) and in branch length-based pairwise distances (normalized BSD). In particular, across subsets, our method achieves the lowest normalized RF distance and the lowest normalized BSD, highlighting that jointly selecting supported completion subtrees and optimizing continuous attachment points is critical for accurate set-wide completion.

Our study opens several new directions for future work. First, the integration of the FPT exact-cover routine with frequency thresholds could further reduce the number of insertions. Second, alternative weighting and scaling schemes may improve branch-length calibration. Third, extending the algorithm to unrooted trees would broaden its applicability. Finally, exploring global objectives that coordinate multiple insertions simultaneously may further improve tree set completion accuracy.

\bmhead{Acknowledgements}

The authors would like to thank the Department of Computer Science, University of Sherbrooke, Quebec, Canada, for providing the necessary resources to conduct this research.

\section*{Declarations}

\bmhead{Funding}
This research was funded by the Natural Sciences and Engineering Research Council of Canada - Discovery Grants [\#RGPIN-2022-04322], Canada Graduate Scholarship – Doctoral [\#CGSD-589644-2024], and Fonds de Recherche du Québec - Nature and Technologies [\#326911].

\bmhead{Conflict of Interest}
The authors declare that they have no conflict of interest.

\bmhead{Data availability}
The biological datasets analyzed in this study were derived from the VertLife resource. The prepared datasets used in this evaluation are publicly available in the project repository at: \url{https://github.com/tahiri-lab/overlap-treeset-completion/evaluation/data/}

\bmhead{Code availability}
The open-source Python implementation of the proposed algorithm, together with the scripts used for dataset preparation and evaluation, is publicly available at: \url{https://github.com/tahiri-lab/overlap-treeset-completion/}


\bibliography{sn-bibliography}% common bib file

\end{document}